\documentclass[11pt]{article}
\usepackage[letterpaper, margin=1in]{geometry}
\usepackage[numbers,sort&compress]{natbib}
\usepackage[USenglish]{babel}
\usepackage[utf8]{inputenc} % allow utf-8 input
\usepackage[T1]{fontenc}    % use 8-bit T1 fonts
\usepackage{url}            % simple URL typesetting
\usepackage{booktabs}       % professional-quality tables
\usepackage{amsfonts}       % blackboard math symbols
\usepackage{nicefrac}       % compact symbols for 1/2, etc.
\usepackage{microtype}      % microtypography
\usepackage{svg}
\usepackage{amssymb, amsmath, amsthm, bbm}
\usepackage{todonotes}
\usepackage{float}
\usepackage{placeins}
\usepackage{subfig}
\usepackage{mathtools}
\usepackage{algorithm}
\usepackage[noend]{algpseudocode}
\usepackage{makecell}
\usepackage{xspace}
\usepackage{sidecap}
\usepackage{multicol}
\usepackage{graphicx}
\usepackage{hyperref}
\usepackage[nameinlink,capitalise]{cleveref} % cleverref after almost everything

\title{Approximating $(k,\ell)$-Median Clustering for Polygonal Curves}
\author{
    Maike Buchin\thanks{
    Faculty of Mathematics,
    Ruhr-University Bochum, Germany,
    \texttt{maike.buchin@rub.de}}
    \and
    Anne Driemel\thanks{
    Hausdorff Center for Mathematics,
    University of Bonn, Germany,
    \texttt{driemel@cs.uni-bonn.de}}
    \and
    Dennis Rohde\thanks{
    Faculty of Mathematics,
    Ruhr-University Bochum, Germany,
    \texttt{dennis.rohde-t1b@rub.de}}
}
\date{\today}

\newtheorem{definition}{Definition}[section]
\crefname{definition}{Definition}{Definitions}
\newtheorem{lemma}[definition]{Lemma}
\crefname{lemma}{Lemma}{Lemmata}

\crefname{proposition}{Proposition}{Propositions}
\newtheorem{theorem}[definition]{Theorem}
\crefname{theorem}{Theorem}{Theorems}

\crefname{observation}{Observation}{Observations}
\newtheorem{corollary}[definition]{Corollary}
\crefname{corollary}{Corollary}{Corollaries}

\DeclarePairedDelimiter{\parens}{\lparen}{\rparen}

\DeclareMathOperator*{\argmin}{arg\,min}

\DeclareMathOperator{\pcost}{cost}

\DeclareMathOperator{\pexpected}{E}

\DeclareMathOperator{\psimpl}{simpl}
\DeclareMathOperator{\ppoly}{poly}

\newcommand{\poly}[1]{\ensuremath{\ppoly\parens*{#1}}}
\newcommand{\expected}[1]{\ensuremath{\pexpected\left[#1\right]}}
\newcommand{\conferre}{cf.~}

\newcommand{\cost}[2]{\ensuremath{\pcost\parens*{#1,#2}}}

\newcommand{\simpl}[2]{\ensuremath{\psimpl\parens*{#1,#2}}}
\let\epsilon\relax
\newcommand{\epsilon}{\varepsilon}

\allowdisplaybreaks

\begin{document}
\maketitle
\thispagestyle{empty}
\setlength{\parindent}{0pt}
\begin{abstract}
    In 2015, Driemel, Krivošija and Sohler introduced the $(k,\ell)$-median problem for clustering polygonal curves under the Fr\'echet distance. Given a set of input curves, the problem asks to find $k$ median curves of at most $\ell$ vertices each that minimize the sum of Fr\'echet distances over all input curves to their closest median curve. A major shortcoming of their algorithm is that the input curves are restricted to lie on the real line.
    In this paper, we present a randomized bicriteria-approximation algorithm that works for polygonal curves in $\mathbb{R}^d$ and achieves approximation factor $(1+\epsilon)$ with respect to the clustering costs. The algorithm has worst-case running-time linear in the number of curves, polynomial in the maximum number of vertices per curve, i.e. their complexity, and exponential in $d$, $\ell$, $\epsilon$ and $\delta$, i.e.,~the failure probability. We achieve this result through a shortcutting lemma, which guarantees the existence of a polygonal curve with similar cost as an optimal median curve of complexity $\ell$, but of complexity at most $2\ell-2$, and whose vertices can be computed efficiently. We combine this lemma with the superset-sampling technique by Kumar et al. to derive our clustering result.
    In doing so, we describe and analyze a generalization of the algorithm by Ackermann et al., which may be of independent interest.
\end{abstract}
\clearpage

\setcounter{page}{1}

\section{Introduction}
Since the development of $k$-means -- the pioneer of modern computational clustering -- the last 65 years have brought a diversity of specialized \cite{SCHAEFFER200727, DBLP:journals/jmlr/Ben-HurHSV01, DBLP:conf/stoc/Har-PeledM04, DBLP:journals/ml/BansalBC04, DBLP:journals/tit/CilibrasiV05,  DBLP:books/sp/16/GuhaM16, DBLP:journals/spm/Vidal11} as well as generalized clustering algorithms \cite{Johnson1967, DBLP:journals/talg/AckermannBS10, DBLP:journals/jmlr/BanerjeeMDG05}. However, in most cases clustering of point sets was studied. Many clustering problems indeed reduce to clustering of point sets, but for sequential data like time-series and trajectories -- which arise in the natural sciences, medicine, sports, finance, ecology, audio/speech analysis, handwriting and many more -- this is not the case. Hence, we need specialized clustering methods for these purposes, \conferre \cite{doi:10.1111/1467-9469.00350, doi:10.1111/j.1467-9868.2007.00605.x, RePEc:spr:jclass:v:22:y:2005:i:2:p:185-201, PETITJEAN201276, PETITJEAN2011678}. 

A promising branch of this active research deals with $(k,\ell)$-center and $(k,\ell)$-median clustering -- adaptions of the well-known Euclidean $k$-center and $k$-median clustering. In $(k,\ell)$-center clustering, respective $(k,\ell)$-median clustering, we are given a set of $n$ polygonal curves in $\mathbb{R}^d$ of complexity (i.e.,~the number of vertices of the curve) at most $m$ each and want to compute $k$ centers that minimize the objective function -- just as in Euclidean $k$-clustering. In addition, the centers are restricted to have complexity at most $\ell$ each to prevent over-fitting -- a problem specific for sequential data. 
A great benefit of regarding the sequential data as polygonal curves is that we introduce an implicit linear interpolation. This does not require any additional storage space since we only need to store the vertices of the curves, which are the sequences at hand. We compare the polygonal curves by their Fr\'echet distance, that is a continuous distance measure which takes the entire course of the curves into account, not only the pairwise distances among their vertices. Therefore, irregular sampled sequences are automatically handled by the interpolation, which is desirable in many cases. 
Moreover,  \citet{DBLP:conf/gis/BuchinDLN19} showed, by using heuristics, that the $(k,\ell)$-clustering objectives yield promising results on trajectory data.

This branch of research formed only recently, about twenty years after \citeauthor{alt_godau} developed an algorithm to compute the Fr\'echet distance between polygonal curves \cite{alt_godau}. Several papers have since studied this type of clustering~\cite{DBLP:conf/soda/DriemelKS16, k_l_center, DBLP:conf/gis/BuchinDLN19, DBLP:conf/swat/BuchinDS20, DBLP:conf/nips/MeintrupMR19}. However, all of these clustering algorithms, except the approximation-schemes for polygonal curves in $\mathbb{R}$ \citep{DBLP:conf/soda/DriemelKS16} and the heuristics in \citep{DBLP:conf/gis/BuchinDLN19}, choose a $k$-subset of the input as centers. (This is also often called \emph{discrete} clustering.) This $k$-subset is later simplified, or all input-curves are simplified before choosing a $k$-subset. 
Either way, using these techniques one cannot achieve an approximation factor of less than $2$. This is because there need not be an input curve with distance to its median which is less than the average distance of a curve to its median.

\citet{DBLP:conf/soda/DriemelKS16}, who were the first to study clustering of polygonal curves under the Fr\'echet distance in this setting, already overcame this problem in one dimension by defining and analyzing $\delta$-signatures, which are succinct representations of classes of curves that allow synthetic center-curves to be constructed. However, it seems that $\delta$-signatures are only  applicable in $\mathbb{R}$. 
Here, we extend their work and obtain the first randomized bicriteria approximation algorithm for $(k,\ell)$-median clustering of polygonal curves in $\mathbb{R}^d$.

\subsection{Related Work}

\citet{DBLP:conf/soda/DriemelKS16} introduced the $(k,\ell)$-center and $(k,\ell)$-median objectives and developed the first approximation-schemes for these objectives, for curves in $\mathbb{R}$. Furthermore, they proved that $(k,\ell)$-center as well as $(k,\ell)$-median clustering is NP-hard, where $k$ is a part of the input and $\ell$ is fixed. Also, they showed that the doubling dimension of the metric space of polygonal curves under the Fr\'echet distance is unbounded, even when the complexity of the curves is bounded.  

Following this work, \citet{k_l_center} developed a constant-factor approximation algorithm for $(k,\ell)$-center clustering in $\mathbb{R}^d$. Furthermore, they provide improved results on the hardness of approximating $(k,\ell)$-center clustering under the Fr\'echet distance: the $(k,\ell)$-center problem is NP-hard to approximate within a factor of $(1.5 - \epsilon)$ for curves in $\mathbb{R}$ and within a factor of $(2.25 - \epsilon)$ for curves in $\mathbb{R}^d$, where $d \geq 2$, in both cases even if $k = 1$.
Furthermore, for the $(k,\ell)$-median variant, \citet{DBLP:conf/swat/BuchinDS20} proved NP-hardness using a similar reduction. Again, the hardness holds even if $k=1$. Also, they provided $(1+\epsilon)$-approximation algorithms for $(k,\ell)$-center, as well as $(k,\ell)$-median clustering, under the discrete Fr\'echet distance. \citet{abhin2020kmedian} give improved algorithms for $(1+\epsilon)$-approximation of $(k,\ell)$-median clustering under discrete Fr\'echet and Hausdorff distance. Recently, \citet{DBLP:conf/nips/MeintrupMR19} introduced a practical $(1+\epsilon)$-approximation algorithm for discrete $k$-median clustering under the Fr\'echet distance, when the input adheres to a certain natural assumption, i.e., the presence of a certain number of outliers.

Our algorithms build upon the clustering algorithm of \citet{10.1109/FOCS.2004.7}, which was later extended by \citet{DBLP:journals/talg/AckermannBS10}. This algorithm is a recursive approximation scheme, that employs two phases in each call. In the so-called \emph{candidate phase} it computes candidates by taking a sample $S$ from the input set $T$ and running an algorithm on each subset of $S$ of a certain size. Which algorithm to use depends on the metric at hand. The idea behind this is simple: if $T$ contains a cluster $T^\prime$ that takes a constant fraction of its size, then a constant fraction of $S$ is from $T^\prime$ with high probability. By brute-force enumeration of all subsets of $S$, we find this subset $S^\prime \subseteq T^\prime$ and if $S$ is taken uniformly and independently at random from $T$ then $S^\prime$ is a uniform and independent sample from $T^\prime$. \citeauthor{DBLP:journals/talg/AckermannBS10} proved for various metric and non-metric distance measures, that $S^\prime$ can be used for computing candidates that contain a $(1+\epsilon)$-approximate median for $T^\prime$ with high probability. The algorithm recursively calls itself for each candidate to eventually evaluate these together with the candidates for the remaining clusters. 

The second phase of the algorithm is the so-called \emph{pruning  phase}, where it partitions its input according to the candidates at hand into two sets of equal size: one with the smaller distances to the candidates and one with the larger distances to the candidates. It then recursively calls itself with the second set as input. The idea behind this is that small clusters now become large enough to find candidates for these. Furthermore, the partitioning yields a provably small error. Finally it returns the set of $k$ candidates that together evaluated best.

\subsection{Our Contributions}

\begin{figure}
    \centering
    \includegraphics[width=0.9\textwidth]{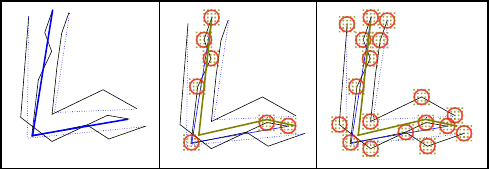}
    \caption{From left to right: symbolic depiction of the operation principle of \cref{alg:1l_median_6,alg:1l_median_3_candidates,alg:1l_median_1_candidates}. Among all approximate $\ell$-simplifications (depicted in blue) of the input curves (depicted in black), \cref{alg:1l_median_6} returns the one that evaluates best (the solid curve) with respect to a sample of the input. \cref{alg:1l_median_3_candidates} does not return a single curve, but a set of candidates. These include the curve returned by \cref{alg:1l_median_6} plus all curves with $\ell$ vertices from the cubic grids, covering balls of certain radius centered at the vertices of an input curve that is close to a median, w.h.p. \cref{alg:1l_median_1_candidates} is similar to \cref{alg:1l_median_3_candidates} but does not only cover the vertices of a single curve, but of multiple curves. We depict the best approximate median that can be generated from the grids in solid green.}
    \label{fig:algorithms_operation_principle}
\end{figure}

We present several algorithms for approximating $(1,\ell)$-median clustering of polygonal curves under the Fr\'echet distance, see \cref{fig:algorithms_operation_principle} for an illustration of the operation principles of our algorithms. While the first one, \cref{alg:1l_median_6}, yields only a coarse approximation (factor 34), it is suitable as plugin for the following two algorithms, \cref{alg:1l_median_3_candidates,alg:1l_median_1_candidates}, due to its asymptotically fast running-time. These algorithms yield a better approximation (factor $3+\epsilon$, respectively  $1+\epsilon$). Additionally, \cref{alg:1l_median_3_candidates,alg:1l_median_1_candidates} are not only able to yield an approximation for the input set $T$, but for a cluster $T^\prime \subseteq T$, that takes a constant fraction of $T$. We would like to use these as plugins to the $(1+\epsilon)$-approximation algorithm for $k$-median clustering by \citet{DBLP:journals/talg/AckermannBS10}, but that would require our algorithms to comply with the sampling properties. For an input set $T$ the weak sampling property expresses that a constant-size set of candidates can be computed, that contains a $(1+\epsilon)$-approximate median for $T$ with high probability, by taking a constant-size uniform and independent sample of $T$. Further, the running-time for computing the candidates depends only on the size of the sample, the size of the candidate set and the failure probability parameter. The strong sampling property is defined similarly, but instead of a candidate set, an approximate median can be computed directly and the running-time may only depend on the size of the sample. In our algorithms, the running-time for computing the candidate set depends on $m$ which is a parameter of the input. Additionally, our first algorithm for computing candidates, which contain a $(3+\epsilon)$-approximate $(1,\ell)$-median with high probability, does not achieve the required approximation-factor of $(1+\epsilon)$. However, looking into the analysis of \citeauthor{DBLP:journals/talg/AckermannBS10}, \emph{any} algorithm for computing candidates, with some guaranteed approximation-factor, can be used in the recursive approximation-scheme. Therefore, we decided to generalize the $k$-median clustering algorithm of \citet{DBLP:journals/talg/AckermannBS10}.

\citet{abhin2020kmedian} use a similar approach, but they developed yet another way to compute candidates: they define and analyze $g$-coverability, which is a generalization of the notion of doubling dimension and indeed, for the discrete Fr\'echet distance the proof builds upon the doubling dimension of points in $\mathbb{R}^d$. However, the doubling dimension of polygonal curves under the Fr\'echet distance is unbounded, even when the complexities of the curves are bounded and it is an open question whether $g$-coverability holds for the continuous Fr\'echet distance. 

We circumvent this by taking a different approach using the idea of shortcutting. It is well-known that shortcutting a polygonal curve (that is, replacing a subcurve by the line segment connecting its endpoints) does not increase its Fr\'echet distance to a line segment. This idea has been used before for a variety of Fr\'echet-distance related problems~\cite{near_linear_curve_simp, DBLP:conf/soda/DriemelKS16, DBLP:journals/siamcomp/DriemelH13,DBLP:journals/comgeo/BuchinBW08}.
Specifically, we introduce two new shortcutting lemmata. These lemmata guarantee the existence of good approximate medians, with complexity at most $2\ell-2$ and whose vertices can be computed efficiently. The first one enables us to return candidates, which contain a $(3+\epsilon)$-approximate median for a cluster inside the input, that takes a constant fraction of the input, w.h.p., and we call it \emph{simple shortcutting}. The second one enables us to return candidates, which contain a $(1+\epsilon)$-approximate median for a cluster inside the input, that takes a constant fraction of the input, w.h.p., and we call it \emph{advanced shortcutting}.
All in all, we obtain as main result, following from \cref{coro:1_approx}:

\begin{theorem}
    Given a set $T$ of $n$ polygonal curves in $\mathbb{R}^d$, of complexity at most $m$ each, parameter values $\epsilon \in (0, 0.158]$ and $\delta \in (0,1)$, and constants $k,\ell \in \mathbb{N}$, there exists an algorithm, which computes a set $C$ of $k$ polygonal curves, each of complexity at most $2\ell-2$, such that with probability at least $(1-\delta)$, it holds that
    \begin{align*}
        \cost{T}{C} = \sum_{\tau \in T} \min_{c \in C} d_F(c,\tau) \leq (1+\epsilon) \sum_{\tau \in T} \min_{c \in C^{\ast}} d_F(c,\tau) = (1+\epsilon) \cost{T}{C^\ast}
    \end{align*}
    where $C^{\ast}$ is an optimal $(k,\ell)$-median solution for $T$ under the Fr\'echet distance~$d_F(\cdot, \cdot)$.
    
    The algorithm has worst-case running-time linear in $n$, polynomial in $m$ and exponential in $\delta,\epsilon, d$ and $\ell$.
\end{theorem}

\subsection{Organization}

The paper is organized as follows.
 First we present a simple and fast $34$-approximation algorithm for $(1,\ell)$-median clustering.
 Then, we present the $(3+\epsilon)$-approximation algorithm for $(1,\ell)$-median clustering of a cluster inside the input, that takes a constant fraction of the input, which builds upon simple shortcutting and the $34$-approximation algorithm. Then, we present a more practical modification of the $(3+\epsilon)$-approximation algorithm, which achieves a $(5+\epsilon)$-approximation for $(1,\ell)$-median clustering. Following this, we present the similar but more involved $(1+\epsilon)$-approximation algorithm for $(1,\ell)$-median clustering of a cluster inside the input, that takes a constant fraction of the input, which builds upon the advanced shortcutting and the $34$-approximation algorithm.
 Finally we present the generalized recursive $k$-median approximation-scheme, which leads to our main result.

\section{Preliminaries}

Here we introduce all necessary definitions. In the following $d \in \mathbb{N}$ is an arbitrary constant. By $\lVert \cdot \rVert$ we denote the Euclidean norm and for $p \in \mathbb{R}^d$ and $r \in \mathbb{R}_{\geq 0}$ we denote by $B(p,r) = \{ q \in \mathbb{R}^d \mid \lVert p - q \rVert \leq r\}$ the closed ball of radius $r$ with center $p$. By $S_n$ we denote the symmetric group of degree $n$. We give a standard definition of grids:

\begin{definition}[grid]
    \label{def:grid}
    Given a number $r \in \mathbb{R}_{>0}$, for $(p_1, \dots, p_d) \in \mathbb{R}^d$ we define by $G(p,r) = (\lfloor p_1 / r \rfloor \cdot r, \dots, \lfloor p_d / r \rfloor \cdot r)$ the $r$-grid-point of $p$. Let $X \subseteq \mathbb{R}^d$ be a subset of $\mathbb{R}^d$. The grid of cell width $r$ that covers $X$ is the set $\mathbb{G}(X,r) = \{G(p,r) \mid p \in X\}$.
\end{definition}
Such a grid partitions the set $X$ into cubic regions and for each $r \in \mathbb{R}_{>0}$ and $p \in X$ we have that $\lVert p - G(p,r) \rVert \leq \frac{\sqrt{d}}{2} r$. We give a standard definition of polygonal curves:

\begin{definition}[polygonal curve]
    \label{def:polygonal_curve}
	A (parameterized) curve is a continuous mapping $\tau \colon [0,1] \rightarrow \mathbb{R}^d$. A curve $\tau$ is polygonal, iff there exist $v_1, \dots, v_m \in \mathbb{R}^d$, no three consecutive on a line, called $\tau$'s vertices and $t_1, \dots, t_m \in [0,1]$ with $t_1 < \dots < t_m$, $t_1 = 0$ and $t_m = 1$, called $\tau$'s instants, such that $\tau$ connects every two contiguous vertices $v_i = \tau(t_i), v_{i+1} = \tau(t_{i+1})$ by a line segment.
\end{definition}
We call the line segments $\overline{v_1v_2}, \dots, \overline{v_{m-1}v_m}$ the edges of $\tau$ and $m$ the complexity of $\tau$, denoted by $\lvert \tau \rvert$. Sometimes we will argue about a sub-curve $\tau$ of a given curve $\sigma$. We will then refer to $\tau$ by restricting the domain of $\sigma$, denoted by $\sigma\vert_X$, where $X \subseteq [0,1]$. 

\begin{definition}[Fr\'echet distance]
    \label{def:frechet_distance}
    Let $\mathcal{H}$ denote the set of all continuous bijections $h\colon [0,1] \rightarrow [0,1]$ with $h(0) = 0$ and $h(1) = 1$, which we call reparameterizations.
    The Fr\'echet distance between curves $\sigma$ and $\tau$ is defined as \[d_F(\sigma, \tau)\ =\ \inf_{h \in \mathcal{H}}\  \max_{t \in [0,1]}\ \lVert \sigma(t) - \tau(h(t)) \rVert. \]
\end{definition}
Sometimes, given two curves $\sigma, \tau$, we will refer to an $h \in \mathcal{H}$ as matching between $\sigma$ and $\tau$.

Note that there must not exist a matching $h \in \mathcal{H}$, such that $\max_{t \in [0,1]} \lVert \sigma(t) - \tau(h(t)) \rVert = d_F(\sigma, \tau)$. This is due to the fact that in some cases a matching realizing the Fr\'echet distance would need to match multiple points $p_1, \dots, p_n$ on $\tau$ to a single point $q$ on $\sigma$, which is not possible since matchings need to be bijections, but the $p_1, \dots, p_n$ can get matched arbitrarily close to $q$, realizing $d_F(\sigma, \tau)$ in the limit, which we formalize in the following lemma:

\begin{lemma}
    \label{lem:sequence}
    Let $\sigma, \tau \colon [0,1] \rightarrow \mathbb{R}^d$ be curves. Let $r = d_F(\sigma, \tau)$. There exists a sequence $(h_i)_{i=1}^\infty$ in $\mathcal{H}$, such that $\lim\limits_{i \to \infty} \max\limits_{t \in [0,1]} \lVert \sigma(t) - \tau(h_i(t)) \rVert = r$.
\end{lemma}
\begin{proof}
    Define $\rho \colon \mathcal{H} \rightarrow \mathbb{R}_{\geq 0}, h \mapsto \max\limits_{t\in [0,1]} \lVert \sigma(t) - \tau(h(t)) \rVert$ with image $R = \{ \rho(h) \mid h \in \mathcal{H} \}$. Per definition, we have $d_F(\sigma, \tau) = \inf R = r$.
    
    For any non-empty subset $X$ of $\mathbb{R}$ that is bounded from below and for every $\epsilon > 0$ it holds that there exists an $x \in X$ with $\inf X \leq x < \inf X + \epsilon$, by definition of the infimum. Since $R \subseteq \mathbb{R}$ and $\inf R$ exists, for every $\epsilon > 0$ there exists an $r^\prime \in R$ with $\inf R \leq r^\prime < \inf R + \epsilon$. 
    
    Now, let $a_i = 1/i$ be a zero sequence. For every $i \in \mathbb{N}$ there exists an $r_i \in R$ with $r \leq r_i < r + a_i$, thus $\lim\limits_{i \to \infty} r_i = r$. 
    
    Let $\rho^{-1}(r^\prime) = \{ h \in \mathcal{H} \mid \rho(h) = r^\prime \}$ be the preimage of $\rho$. Since $\rho$ is a function, $\lvert \rho^{-1}(r^\prime) \rvert \geq 1$ for each $r^\prime \in R$. Now, for $i \in \mathbb{N}$, let $h_i$ be an arbitrary element from $\rho^{-1}(r_i)$. By definition it holds that \[ \lim\limits_{i \to \infty} \max\limits_{t \in [0,1]} \lVert \sigma(t) - \tau(h_i(t)) \rVert = \lim_{i \to \infty} \rho(h_i) = \lim_{i \to \infty} r_i = r = \inf R, \] which proves the claim.
\end{proof}

Now we introduce the classes of curves we are interested in.

\begin{definition}[polygonal curve classes]
    For $d \in \mathbb{N}$, we define by $\mathbb{X}^d$ the equivalence class of polygonal curves (where two curves are equivalent, iff they can be made identical by a reparameterization) in ambient space $\mathbb{R}^d$. For $m \in \mathbb{N}$ we define by $\mathbb{X}^d_m$ the subclass of polygonal curves of complexity at most $m$.
\end{definition}

Simplification is a fundamental problem related to curves and which appears as sub-problem in our algorithms.

\begin{definition}[minimum-error $\ell$-simplification]
    For a polygonal curve $\tau \in \mathbb{X}^d$ we denote by $\simpl{\alpha}{\tau}$ an $\alpha$-approximate minimum-error $\ell$-simplification of $\tau$, i.e., a curve $\sigma \in \mathbb{X}^d_\ell$ with $d_F(\tau, \sigma) \leq \alpha \cdot d_F(\tau, \sigma^\prime)$ for all $\sigma^\prime \in \mathbb{X}^d_\ell$.
\end{definition}

Now we define the $(k,\ell)$-median clustering problem for polygonal curves.

\begin{definition}[$(k,\ell)$-median clustering]
    The $(k,\ell)$-median clustering problem is defined as follows, where $k,l \in \mathbb{N}$ are fixed (constant) parameters of the problem: given a finite and non-empty set $T \subset \mathbb{X}^d_m$ of polygonal curves, compute a set of $k$ curves $C^\ast \subset \mathbb{X}^d_\ell$, such that $\cost{T}{C^\ast} = \sum\limits_{\tau \in T} \min\limits_{c^\ast \in C^\ast} d_F(\tau, c^\ast)$ is minimal.
\end{definition}

We call $\cost{\cdot}{\cdot}$ the objective function and we often write $\cost{T}{c}$ as shorthand for $\cost{T}{\{c\}}$. The following theorem of Indyk~\cite{Indyk00} is useful for evaluating the cost of a curve at hand.

\begin{theorem}{\citep[Theorem 31]{Indyk00}}
	\label{theo:indyk_median}
	Let $\epsilon \in (0,1]$ and $T \subset \mathbb{X}^d$ be a set of polygonal curves. Further let $W$ be a non-empty sample, drawn uniformly and independently at random from $T$, with replacement. For $\tau, \sigma \in T$ with $\cost{T}{\tau} > (1+\epsilon) \cost{T}{\sigma}$ it holds that $\Pr[\cost{W}{\tau} \leq \cost{W}{\sigma}] < \exp\left( - {\epsilon^2 \lvert W \rvert}/{64} \right)$.
\end{theorem}

The following concentration bound also applies to independent Bernoulli trials, which are a special case of Poisson trials where each trial has same probability of success. \citet{10.1109/FOCS.2004.7} use this to bound the probability that a subset $S^\prime$ of an independent and uniform sample $S$ from a set $T$ is entirely contained in a subset $T^\prime$ of $T$. They call it superset-sampling.

\begin{lemma}[Chernoff bound for independent Poisson trials]{\citep[Theorem 4.5]{probability_and_computing}}
    \label{lem:bernoulli_trial_bound}
    Let $X_1, \dots, X_n$ be independent Poisson trials. For $\delta \in (0,1)$ it holds that \[ \Pr\left[\sum_{i=1}^n X_i \leq (1-\delta) \expected{\sum_{i=1}^n X_i} \right] \leq \exp\left( - \frac{\delta^2}{2} \expected{\sum_{i=1}^n X_i} \right). \]
\end{lemma}

\section{Simple and Fast \texorpdfstring{$34$-}{34-}Approximation for \texorpdfstring{$(1,\ell)$-}{(1,l)-}Median}

Here, we present \cref{alg:1l_median_6}, a $34$-approximation algorithm for $(1,\ell)$-median clustering, which is based on the following facts: we can obtain a $3$-approximate solution to the $(1,\ell)$-median for a given set $T = \{ \tau_1, \dots, \tau_n \} \subset \mathbb{X}^d_m$ of polygonal curves in terms of objective value, i.e., we obtain one of the at least $n/2$ input curves that are within distance $2 \cdot \cost{T}{c^\ast}/n$ to an optimal $(1,\ell)$-median $c^\ast$ for $T$, w.h.p., by uniformly and independently sampling a sufficient number of curves from $T$. There are at least $n/2$ of these curves by an averaging argument. These curves have cost up to $3 \cdot \cost{T}{c^\ast}$ by the triangle-inequality. The sample has size depending only on a parameter determining the failure probability and we can improve on running-time even more by using \cref{theo:indyk_median} and evaluate the cost of each curve in the sample of candidates against another sample of similar size instead of against the complete input. Though, we have to accept an approximation factor of $5$ (if we set $\epsilon = 1$ in \cref{theo:indyk_median}). That is indeed acceptable, since we only obtain an approximate solution in terms of objective value and completely ignore the bound on the number of vertices of the center curve, which is a disadvantage of this approach and results in the lower bound of $\cost{T}{c^\ast}$ not necessarily holding (if $\ell < m$). To fix this, we simplify the candidate curve that evaluated best against the second sample, using an efficient minimum-error $\ell$-simplification approximation algorithm, which downgrades the approximation factor to $6+7\alpha$, where $\alpha$ is the approximation factor of the minimum-error $\ell$-simplification.

However, \cref{alg:1l_median_6} is very fast in terms of the input size. Indeed, it has worst-case running-time independent of $n$ and sub-quartic in $m$. Now, \cref{alg:1l_median_6} has the purpose to provide us an approximate median for a given set of polygonal curves: the bi-criteria approximation algorithms (\cref{alg:1l_median_3_candidates,alg:1l_median_1_candidates}), which we present afterwards and which are capable of generating center curves with up to $2\ell-2$ vertices, need an approximate median (and the approximation factor) to bound the optimal objective value. Furthermore, there is a case where \cref{alg:1l_median_3_candidates,alg:1l_median_1_candidates} may fail to provide a good approximation, but it can be proven that the result of \cref{alg:1l_median_6} is then a very good approximation, which can be used instead.

\begin{algorithm}[H]
\caption{$(1,\ell)$-Median by Simplification\label{alg:1l_median_6}}
    \begin{algorithmic}[1]
        \Procedure{$(1,\ell)$-Median-$34$-Approximation}{$T = \{\tau_1, \dots, \tau_n \}$, $\delta$}
            \State $S \gets$ sample $\left\lceil 2 (\ln(2)-\ln(\delta)) \right\rceil$ curves from $T$ uniformly and independently with replacement
            \State $\gamma \gets \left \lceil -64 (\ln(\delta) - \ln(\lceil 4\ln(2)-\ln(\delta) \rceil) \right \rceil$
            \State $W \gets$ sample $\gamma$ curves from $T$ uniformly and independently with replacement
            \State $t \gets$ arbitrary elem. from $\argmin\limits_{s \in S} \cost{W}{s}$
            \State \Return \simpl{\alpha}{t} \Comment{E.g. combining \citep{alt_godau,imai_polygonal_1988}}
        \EndProcedure
    \end{algorithmic}
\end{algorithm}

Next, we prove the quality of approximation of \cref{alg:1l_median_6}.

\begin{theorem}
    \label{theo:6_alpha_approx}
    Given a parameter $\delta \in (0,1)$ and a set $T = \{ \tau_1, \dots, \tau_n \} \subset \mathbb{X}^d_m$ of polygonal curves, \cref{alg:1l_median_6} returns with probability at least $1-\delta$ a polygonal curve $c \in \mathbb{X}^d_\ell$, such that $\cost{T}{c^\ast} \leq \cost{T}{c} \leq (6 + 7\alpha) \cdot \cost{T}{c^\ast}$, where $c^\ast$ is an optimal $(1,\ell)$-median for $T$ and $\alpha$ is the approximation-factor of the utilized minimum-error $\ell$-simplification approximation algorithm.
\end{theorem}
\begin{proof}
    First, we know that $d_F(\tau, \simpl{\alpha}{\tau}) \leq \alpha \cdot d_F(\tau, c^\ast)$, for each $\tau \in T$.

    Now, there are at least $\frac{n}{2}$ curves in $T$ that are within distance at most $\frac{2\cost{T}{c^\ast}}{n}$ to $c^\ast$. Otherwise the cost of the remaining curves would exceed $\cost{T}{c^\ast}$, which is a contradiction. Hence each $s \in S$ has probability at least $\frac{1}{2}$ to be within distance $\frac{2\cost{T}{c^\ast}}{n}$ to $c^\ast$.
    
    Since the elements of $S$ are sampled independently we conclude that the probability that every $s \in S$ has distance to $c^\ast$ greater than $\frac{2\cost{T}{c^\ast}}{n}$ is at most $(1-\frac{1}{2})^{\lvert S \rvert} \leq \exp\left(-\frac{2(\ln(2)-\ln(\delta))}{2}\right) = \frac{\delta}{2}$.
    
    Now, assume there is a $s \in S$ with $d_F(s, c^\ast) \leq \frac{2 \cost{T}{c^\ast}}{n}$. We do not want any $t \in S \setminus \{s\}$ with $\cost{T}{t} > 2 \cost{T}{s}$ to have $\cost{W}{t} \leq \cost{W}{s}$. Using \cref{theo:indyk_median} we conclude that this happens with probability at most \[\exp\left(-\frac{-64 (\ln(\delta) - \ln(\lceil 4\ln(2)-\ln(\delta)\rceil)}{64}\right) \leq \frac{\delta}{\lceil 4 (\ln(2) - \ln(\delta)) \rceil} \leq \frac{\delta}{2 \lvert S \rvert},\] for each $t \in S \setminus \{s\}$.
    
    Using a union bound over all bad events, we conclude that with probability at least $1-\delta$, \cref{alg:1l_median_6} samples a curve $s \in S$, with $d_F(s, c^\ast) \leq 2 \cost{T}{c^\ast}/n$ and returns the simplification $c = \simpl{\alpha}{t}$ of a curve $t \in S$, with $\cost{T}{t} \leq 2 \cost{T}{s}$. The triangle-inequality yields
    \[\sum_{\tau \in T} (d_F(t, c^\ast) - d_F(\tau, c^\ast)) \leq \sum_{\tau \in T} d_F(t, \tau) \leq 2 \sum_{\tau \in T} d_F(s, \tau) \leq 2 \sum_{\tau \in T} (d_F(\tau, c^\ast) + d_F(c^\ast, s)),\]
    which is equivalent to
    \[ n \cdot d_F(t, c^\ast) \leq 2 \cost{T}{c^\ast} + \cost{T}{c^\ast} + 2 n \frac{2\cost{T}{c^\ast}}{n} \Leftrightarrow{} d_F(t, c^\ast) \leq \frac{7\cost{T}{c^\ast}}{n}. \]
    
    Hence, we have
    \begin{align*}
        \cost{T}{c} & ={} \sum_{\tau \in T} d_F(\tau, \simpl{\alpha}{t}) \leq \sum_{\tau \in T} (d_F(\tau, t) + d_F(t, \simpl{\alpha}{t})) \\
        & \leq{} 2 \cost{T}{s} + \sum_{\tau \in T} \alpha \cdot d_F(t, c^\ast) \leq{} 2 \sum_{\tau \in T} (d_F(\tau, c^\ast) + d_F(c^\ast, s)) + 7\alpha \cdot \cost{T}{c^\ast} \\
        & \leq{} 2 \cost{T}{c^\ast} + 4 \cost{T}{c^\ast} + 7\alpha \cdot \cost{T}{c^\ast} ={} (6 + 7\alpha) \cost{T}{c^\ast}.
    \end{align*}
    
    The lower bound $\cost{T}{c^\ast} \leq \cost{T}{c}$ follows from the fact that the returned curve has $\ell$ vertices and that $c^\ast$ has minimum cost among all curves with $\ell$ vertices.
\end{proof}

The following lemma enables us to obtain a concrete approximation-factor and worst-case running-time of \cref{alg:1l_median_6}.

\begin{lemma}[{\citet[Lemma 7.1]{k_l_center}}]
    \label{lem:imai_iri_simpli}
    Given a curve $\sigma \in \mathbb{X}^d_m$, a $4$-approximate minimum-error $\ell$-simplification can be computed in $O(m^3 \log m)$ time.
\end{lemma}

The simplification algorithm used for obtaining this statement is a combination of the algorithm by \citet{imai_polygonal_1988} and the algorithm by \citet{alt_godau}. Combining \cref{theo:6_alpha_approx} and \cref{lem:imai_iri_simpli}, we obtain the following corollary.

\begin{corollary}
    \label{coro:1l_median_34}
    Given a parameter $\delta \in (0,1)$ and a set $T \subset \mathbb{X}^d_m$ of polygonal curves, \cref{alg:1l_median_6} returns with probability at least $1-\delta$ a polygonal curve $c \in \mathbb{X}^d_\ell$, such that $\cost{T}{c^\ast} \leq \cost{T}{c} \leq 34 \cdot \cost{T}{c^\ast}$, where $c^\ast$ is an optimal $(1,\ell)$-median for $T$, in time $O(m^2 \log(m) (-\ln^2 \delta) + m^3 \log m)$, when the algorithms by \citet{imai_polygonal_1988} and \citet{alt_godau} are combined for $\ell$-simplification.
\end{corollary}
\begin{proof}
    We use \cref{lem:imai_iri_simpli} together with \cref{theo:6_alpha_approx}, which yields an approximation factor of $34$. Now, drawing the first sample takes time $O(-\ln \delta)$. Drawing the second sample also takes time $O(-\ln(\delta))$ and evaluating the samples against each other takes time $O(m^2 \log(m) (-\ln^2 \delta))$. Simplifying one of the curves that evaluates best takes time $O(m^3 \log m)$. We conclude that \cref{alg:1l_median_6} has running-time $O(m^2 \log(m) (-\ln^2 \delta) + m^3 \log m)$.
\end{proof}

\section{\texorpdfstring{$(3+\epsilon)$-}{}Approximation for \texorpdfstring{$(1,\ell)$-}{(1,l)-}Median by Simple Shortcutting}

\FloatBarrier

Here, we present \cref{alg:1l_median_3_candidates}, which returns candidates, containing a $(3+\epsilon)$-approximate $(1,\ell)$-median of complexity at most $2\ell-2$, for a cluster contained in the input, that takes a constant fraction of the input, w.h.p. \cref{alg:1l_median_3_candidates} can be used as plugin in our generalized version (\cref{alg:kl_median}, \cref{sec:k-median}) of the algorithm by \citet{DBLP:journals/talg/AckermannBS10}.

In contrast to \citet{abhin2020kmedian} we cannot use the property, that the vertices of a median must be found in the balls of radius $d_F(\tau, c^\ast)$, centered at $\tau$'s vertices, where $c^\ast$ is an optimal $(1,\ell)$-median for a given input $T$, which $\tau$ is an element of. This is an immediate consequence of using the continuous Fr\'echet distance.

We circumvent this by proving the following shortcutting lemmata. We start with the simplest, which states that we can indeed search the aforementioned balls, if we accept a resulting curve of complexity at most $2\ell-2$. See \cref{fig:simple_shortcutting} for a visualization.

\begin{SCfigure}
    \includegraphics[scale=0.8]{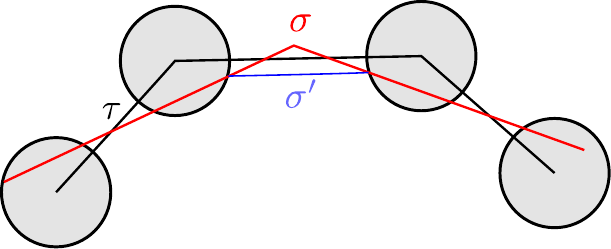}
    \caption{Visualization of a simple shortcut. The black curve is an input-curve that is close to an optimal median, which is depicted in red. By inserting the blue shortcut we can find a curve that has the same distance to the black curve as the median but with all vertices contained in the balls centered at the black curves vertices.}
    \label{fig:simple_shortcutting}
\end{SCfigure}

\begin{lemma}[shortcutting using a single polygonal curve]
    \label{lem:shortcutting_single}
    Let $\sigma, \tau \in \mathbb{X}^d$ be polygonal curves. Let $v^\tau_1, \dots, v^\tau_{\lvert \tau \rvert}$ be the vertices of $\tau$ and let $r = d_F(\sigma, \tau)$. There exists a polygonal curve $\sigma^\prime \in \mathbb{X}^d$ with every vertex contained in at least one of $B(v^\tau_1, r), \dots, B(v^\tau_{\lvert \tau \rvert}, r)$, $d_F(\sigma^\prime, \tau) \leq d_F(\sigma, \tau)$ and $\lvert \sigma^\prime \rvert \leq 2 \lvert \sigma \rvert - 2$.
\end{lemma}
\begin{proof}
    Let $v^\sigma_1, \dots, v^\sigma_{\lvert \sigma \rvert}$ be the vertices of $\sigma$. Further, let $t^\sigma_1, \dots, t^\sigma_{\lvert \sigma \rvert}$ and $t^\tau_1, \dots, t^\tau_{\lvert \tau \rvert}$ be the instants of $\sigma$ and $\tau$, respectively. Also, for $h \in \mathcal{H}$ (recall that $\mathcal{H}$ is the set of all continuous bijections $h\colon [0,1] \rightarrow [0,1]$ with $h(0) = 0$ and $h(1) = 1$), let $r_h = \max\limits_{t \in [0,1]} \lVert \sigma(t) - \tau(h(t)) \rVert$ be the distance realized by $h$. We know from \cref{lem:sequence} that there exists a sequence $(h_x)_{x=1}^\infty$ in $\mathcal{H}$, such that $\lim\limits_{x \to \infty} r_{h_x} = d_F(\sigma, \tau) = r$.
    
    Now, fix an arbitrary $h \in \mathcal{H}$ and assume there is a vertex $v^\sigma_i$ of $\sigma$, with instant $t^\sigma_i$, that is not contained in any of $B(v^\tau_1, r_h), \dots, B(v^\tau_{\lvert \tau \rvert}, r_h)$. Let $j$ be the maximum of $\{1, \dots, \lvert \tau \rvert - 1\}$, such that $t^\tau_j \leq h(t^\sigma_i) \leq t^\tau_{j+1}$. So $v^\sigma$ is matched to $\overline{\tau(t^\tau_j) \tau(t^\tau_{j+1})}$ by $h$. We modify $\sigma$ in such a way, that $v^\sigma_i$ is replaced by two new vertices that are elements of $B(v^\tau_j, r_h)$ and $B(v^\tau_{j+1}, r_h)$, respectively.
    
    Namely, let $t^-$ be the maximum of $[0, t^\sigma_i)$, such that $\sigma(t^-) \in B(v^\tau_j, r_h)$ and let $t^+$ be the minimum of $(t^\sigma_i, 1]$, such that $\sigma(t^+) \in B(v^\tau_{j+1}, r_h)$. These are the instants when $\sigma$ leaves $B(v^\tau_j, r_h)$ before visiting $v^\sigma_i$ and $\sigma$ enters $B(v^\tau_{j+1}, r_h)$ after visiting $v^\sigma_i$, respectively. Let $\sigma^\prime_h$ be the piecewise defined curve, defined just like $\sigma$ on $[0,t^-]$ and $[t^+,1]$, but on $(t^-, t^+)$ it connects $\sigma(t^-)$ and $\sigma(t^+)$ with the line segment $s(t) = \left(1-\frac{t-t^-}{t^+-t^-}\right) \tau(t^-) + \frac{t-t^-}{t^+-t^-} \tau(t^+)$. 
    
    We know that $\lVert \sigma(t^-) - \tau(h(t^-)) \rVert \leq r_h$ and $\lVert \sigma(t^+) - \tau(h(t^+)) \rVert \leq r_h$. Note that $t^\tau_j \leq h(t^-)$ and $h(t^+) \leq t^\tau_{j+1}$ since $\sigma(t^-)$ and $\sigma(t^+)$ are closest points to $v^\sigma_i$ on $\sigma$ that have distance $r_h$ to $v^\tau_j$ and $v^\tau_{j+1}$, respectively, by definition. Therefore, $\tau$ has no vertices between the instants $h(t^-)$ and $h(t^+)$. Now, $h$ can be used to match $\sigma^\prime_h\vert_{[0,t^-)}$ to $\tau\vert_{[0,h(t^-))}$ and $\sigma^\prime_h\vert_{(t^+,1]}$ to $\tau\vert_{(t^+,1]}$ with distance at most $r_h$. Since $\sigma^\prime_h\vert_{[t^-, t^+]}$ and $\tau\vert_{[h(t^-), h(t^+)]}$ are just line segments, they can be matched to each other with distance at most $\max\{\lVert \sigma^\prime_h(t^-) - \tau(h(t^-)) \rVert, \lVert \sigma^\prime_h(t^+) - \tau(h(t^+)) \rVert \} \leq r_h$. We conclude that $d_F(\sigma^\prime_h, \tau) \leq r_h$.
    
    Because this modification works for every $h \in \mathcal{H}$, we have $d_F(\sigma^\prime_{h}, \tau) \leq r_h$ for every $h \in \mathcal{H}$. Thus, $\lim\limits_{x \to \infty} d_F(\sigma^\prime_{h_x}, \tau) \leq d_F(\sigma, \tau) = r$.
    
    Now, to prove the claim, for every $h \in \mathcal{H}$ we apply this modification to $v^\sigma_i$ and successively to every other vertex $v^{\sigma^\prime_h}_i$ of the resulting curve $\sigma^\prime_h$, not contained in one of the balls, until every vertex of $\sigma^\prime_h$ is contained in a ball. Note that the modification is repeated at most $\lvert \sigma \rvert - 2$ times for every $h \in \mathcal{H}$, since the start and end vertex of $\sigma$ must be contained in $B(v^\tau_1, r_h)$ and $B(v^\tau_{\lvert \tau \rvert}, r_h)$, respectively. Therefore, the number of vertices of every $\sigma^\prime_h$ can be bounded by $2 \cdot (\lvert \sigma \rvert - 2) + 2$ since every other vertex must not lie in a ball and for each such vertex one new vertex is created. Thus, $\lvert \sigma^\prime_h \rvert \leq 2 \lvert \sigma \rvert - 2$.
\end{proof}

We now present \cref{alg:1l_median_3_candidates}, which works similar as \cref{alg:1l_median_6}, but uses shortcutting instead of simplification. As a consequence, we can achieve an approximation factor of $3+\epsilon$, instead of a factor of $2+\epsilon+\alpha$, where $\alpha \geq 1$ is the approximation factor of the simplifiaction algorithm used in \cref{alg:1l_median_6}. To achieve an approximation-factor of $3+\epsilon$ using simplification, one would need to compute the optimal minimum-error $\ell$-simplifications of the input curves and to the best of our knowledge, there is no such algorithm for the continuous Fr\'echet distance.

In contrast to \cref{alg:1l_median_6}, \cref{alg:1l_median_3_candidates} utilizes the superset-sampling technique by \citet{10.1109/FOCS.2004.7}, i.e., the concentration bound in \cref{lem:bernoulli_trial_bound}, to obtain an approximate $(1,\ell)$-median for a cluster $T^\prime$ contained in the input $T$, that takes a constant fraction of $T$. Therefore, it has running-time exponential in the size of the sample $S$. A further difference is that we need an upper and a lower bound on the cost of an optimal $(1,\ell)$-median for $T^\prime$, to properly set up the grids we use for shortcutting. The lower bound can be obtained by simple estimation, using Markov's inequality. For the upper bound we utilize a case distinction, which guarantees us that if we fail to obtain an upper bound on the optimal cost, the result of \cref{alg:1l_median_6} then is a good approximation (factor $2+\epsilon$) and can be used instead of a best curve obtained by shortcutting.

\cref{alg:1l_median_3_candidates} has several parameters: $\beta$ determines the size (in terms of a fraction of the input) of the smallest cluster inside the input for which an approximate median can be computed, $\delta$ determines the probability of failure of the algorithm and $\epsilon$ determines the approximation factor.

\begin{algorithm}[H]
\caption{$(1,\ell)$-Median for Subset by Simple Shortcutting\label{alg:1l_median_3_candidates}}
    \begin{algorithmic}[1]
        \Procedure{$(1,\ell)$-Median-$(3+\epsilon)$-Candidates}{$T = \{ \tau_1, \dots, \tau_n \}, \beta, \delta, \epsilon$}
            \State $\epsilon^\prime \gets \epsilon/3$, $C \gets \emptyset$
            \State $S \gets $ sample $\left\lceil -8\beta(\epsilon^\prime)^{-1} (\ln(\delta)-\ln(4)) \right\rceil$ curves from $T$ uniformly and independently
            
            \hspace{\algorithmicindent} with replacement
            \For{$S^\prime \subseteq S$ with $\lvert S^\prime \rvert = \frac{\lvert S \rvert}{2 \beta}$}
                \State $c \gets$ $(1,\ell)$-\textsc{Median}-$34$-\textsc{Approximation}$(S^\prime, \delta/4)$ \Comment{\cref{alg:1l_median_6}}
                \State $\Delta \gets \cost{S^\prime}{c}$, $\Delta_l \gets \frac{\delta n}{2 \lvert S \rvert} \frac{\Delta}{34}$, $\Delta_u \gets \frac{1}{\epsilon^\prime} \Delta$, $C \gets C \cup \{c\}$
                \For{$s \in S^\prime$}
                    \State $P \gets \emptyset$
                    \For{$i \in \{1, \dots, \lvert s \rvert\}$}
                        \State $P \gets P \cup \mathbb{G}\left(B\left(v^s_i, (1+\epsilon^\prime)\Delta_u\right),\frac{2 \epsilon^\prime}{n\sqrt{d}}\Delta_l\right)$ \Comment{$v^s_i$: $i$\textsuperscript{th} vertex of $s$}
                    \EndFor
                    \State $C \gets C\ \cup$ set of all polygonal curves with $2\ell-2$ vertices from $P$
                \EndFor
            \EndFor
            \State \Return $C$
        \EndProcedure
    \end{algorithmic}
\end{algorithm}

We prove the quality of approximation of \cref{alg:1l_median_3_candidates}.

\begin{theorem}
    \label{theo:1l_median_3_guarantee}
    Given three parameters $\beta \in [1, \infty)$, $\delta, \epsilon \in (0,1)$ and a set $T = \{ \tau_1, \dots, \tau_n \} \subset \mathbb{X}^d_m$ of polygonal curves, with probability at least $1-\delta$ the set of candidates that \cref{alg:1l_median_3_candidates} returns contains a $(3+\epsilon)$-approximate $(1,\ell)$-median with up to $2\ell-2$ vertices for any $T^\prime \subseteq T$, if $\lvert T^\prime \rvert \geq \frac{1}{\beta} \lvert T \rvert$.
\end{theorem}
\begin{proof}
    We assume that $\lvert T^\prime \rvert \geq \frac{1}{\beta} \lvert T \rvert$. Let $n^\prime$ be the number of sampled curves in $S$ that are elements of $T^\prime$. Clearly, $\expected{n^\prime} \geq \sum_{i=1}^{\lvert S \rvert} \frac{1}{\beta} = \frac{\lvert S \rvert}{\beta}$. Also $n^\prime$ is the sum of independent Bernoulli trials. A Chernoff bound (\conferre \cref{lem:bernoulli_trial_bound}) yields:
    \begin{align*}
        \Pr\left[n^\prime < \frac{\lvert S \rvert}{2\beta}\right] \leq \Pr\left[n^\prime < \frac{1}{2}\expected{n^\prime}\right] \leq \exp\left(-\frac{1}{4}\frac{\lvert S \rvert}{2\beta}\right) \leq \exp\left( \frac{\ln(\delta)-\ln(4)}{\epsilon} \right) = \left(\frac{\delta}{4}\right)^{\frac{1}{\epsilon}} \leq \frac{\delta}{4}.
    \end{align*}
    
    In other words, with probability at most $\delta/4$ no subset $S^\prime \subseteq S$, of cardinality at least $\frac{\lvert S \rvert}{2\beta}$, is a subset of $T^\prime$. We condition the rest of the proof on the contrary event, denoted by $\mathcal{E}_{T^\prime}$, namely, that there is a subset $S^\prime \subseteq S$ with $S^\prime \subseteq T^\prime$ and $\lvert S^\prime \rvert \geq \frac{\lvert S \rvert}{2\beta}$. Note that $S^\prime$ is then a uniform and independent sample of $T^\prime$.

    Now, let $c^\ast \in \argmin\limits_{c \in \mathbb{X}^d_\ell} \cost{T^\prime}{c}$ be an optimal $(1,\ell)$-median for $T^\prime$. The expected distance between $s \in S^\prime$ and $c^\ast$ is \[\expected{d_F(s, c^\ast)\ \vert\ \mathcal{E}_{T^\prime}} = \sum_{\tau \in T^\prime} d_F(c^\ast, \tau) \cdot \frac{1}{\lvert T^\prime \rvert} =  \frac{\cost{T^\prime}{c^\ast}}{\lvert T^\prime \rvert}.\] By linearity we have $\expected{\cost{S^\prime}{c^\ast}\ \vert\ \mathcal{E}_{T^\prime}} = \frac{\lvert S^\prime \rvert}{\lvert T^\prime \rvert} \cost{T^\prime}{c^\ast}$. Markov's inequality yields:
    \begin{align*}
        \Pr\left[ \frac{\delta\lvert T^\prime \rvert}{4\lvert S^\prime \rvert}\cost{S^\prime}{c^\ast} >  \cost{T^\prime}{c^\ast}\ \Big\vert\ \mathcal{E}_{T^\prime}\right] \leq \frac{\delta}{4}.
    \end{align*}
    We conclude that with probability at most $\delta/4$ we have $\frac{\delta \lvert T^\prime \rvert}{4\lvert S^\prime \rvert} \cost{S^\prime}{c^\ast} > \cost{T^\prime}{c^\ast}$.
 
    Using Markov's inequality again, for every $s \in S^\prime$ we have \[\Pr\left[d_F(s, c^\ast) > (1+\epsilon) \frac{\cost{T^\prime}{c^\ast}}{\lvert T^\prime \rvert}\ \Big\vert\ \mathcal{E}_{T^\prime}\right] \leq \frac{1}{1+\epsilon},\] therefore by independence \[\Pr\left[\bigwedge_{s \in S^\prime} \left(d_F(s, c^\ast) > (1+\epsilon)\frac{\cost{T^\prime}{c^\ast}}{\lvert T^\prime \rvert}\right)\ \Big\vert\ \mathcal{E}_{T^\prime}\right] \leq \frac{1}{(1+\epsilon)^{\lvert S^\prime \rvert}} \leq \exp\left(-\frac{\epsilon}{2}\frac{\lvert S \rvert}{2\beta}\right).\] Hence, with probability at most $\exp\left(-\frac{\epsilon \left\lceil -\frac{8\beta(\ln(\delta) - \ln(4))}{\epsilon} \right\rceil}{4\beta}\right) \leq \delta^2/16 \leq \delta/4$ there is no $s \in S^\prime$ with $d_F(s, c^\ast) \leq (1+\epsilon) \frac{\cost{T^\prime}{c^\ast}}{\lvert T^\prime \rvert}$. Also, with probability at most $\delta/4$ \cref{alg:1l_median_6} fails to compute a $34$-approximate $(1,\ell)$-median $c \in \mathbb{X}^d_\ell$ for $S^\prime$, \conferre \cref{coro:1l_median_34}.
    \\
    
    Using a union bound over these bad events, we conclude that with probability at least $1-\delta$ all of the following events occur simultaneously:
    \begin{itemize}
        \item There is a subset $S^\prime \subseteq S$ of cardinality at least $\lvert S \rvert /(2\beta)$ that is a uniform and independent sample of $T^\prime$,
        \item there is a curve $s \in S^\prime$ with $d_F(s, c^\ast) \leq (1+\epsilon)\frac{\cost{T^\prime}{c^\ast}}{\lvert T^\prime \rvert}$,
        \item \cref{alg:1l_median_6} computes a polygonal curve $c \in \mathbb{X}^d_\ell$ with $\cost{S^\prime}{c^\ast_{S^\prime}} \leq \cost{S^\prime}{c} \leq 34 \cost{S^\prime}{c^\ast_{S^\prime}}$, where $c^\ast_{S^\prime} \in \mathbb{X}^d_\ell$ is an optimal $(1,\ell)$-median for $S^\prime$,
        \item and it holds that $\frac{\delta \lvert T^\prime \rvert}{4 \lvert S^\prime \rvert} \cost{S^\prime}{c^\ast} \leq \cost{T^\prime}{c^\ast}$.
    \end{itemize}
    
    Since $c^\ast_{S^\prime}$ is an optimal $(1,\ell)$-median for $S^\prime$ we get the following from the last two items:
    \begin{align*}
        \cost{T^\prime}{c^\ast} \geq \frac{\delta \lvert T^\prime \rvert}{4 \lvert S^\prime \rvert} \cost{S^\prime}{c^\ast} \geq \frac{\delta \lvert T^\prime \rvert}{4 \lvert S^\prime \rvert} \cost{S^\prime}{c^\ast_{S^\prime}} \geq \frac{\delta \lvert T^\prime \rvert}{4 \lvert S^\prime \rvert} \frac{\cost{S^\prime}{c}}{34}.
    \end{align*}
    
    We now distinguish between two cases:
    \\
    
    \textbf{Case 1:} $d_F(c,c^\ast) \geq (1+2\epsilon) \frac{\cost{T^\prime}{c^\ast}}{\lvert T^\prime \rvert}$
    
    The triangle-inequality yields 
    \begin{align*}
        d_F(c,s) & \geq{} d_F(c,c^\ast) - d_F(c^\ast,s) \geq d_F(c,c^\ast) - (1+\epsilon) \frac{\cost{T^\prime}{c^\ast}}{\lvert T^\prime \rvert} \\
        & \geq{} (1+2\epsilon) \frac{\cost{T^\prime}{c^\ast}}{\lvert T^\prime \rvert} - (1+\epsilon) \frac{\cost{T^\prime}{c^\ast}}{\lvert T^\prime \rvert} = \epsilon \frac{\cost{T^\prime}{c^\ast}}{\lvert T^\prime \rvert}.
    \end{align*}
    
    As a consequence, $\cost{S^\prime}{c} \geq \epsilon \frac{\cost{T^\prime}{c^\ast}}{\lvert T^\prime \rvert} \Leftrightarrow \frac{\cost{T^\prime}{c^\ast}}{\lvert T^\prime \rvert} \leq \frac{1}{\epsilon} \cost{S^\prime}{c}$.
    
    Now, let $v^{s}_1, \dots, v^{s}_{\lvert s \rvert}$ be the vertices of $s$. By \cref{lem:shortcutting_single} there exists a polygonal curve $c^\prime$ with up to $2\ell - 2$ vertices, every vertex contained in one of $B(v^{s}_1, d_F(c^\ast, s)), \dots, B(v^{s}_{\lvert s \rvert}, d_F(c^\ast, s))$ and $d_F(s, c^\prime) \leq d_F(s, c^\ast) \leq (1+\epsilon) \frac{\cost{T^\prime}{c^\ast}}{\lvert T^\prime \rvert} \leq (1+\epsilon) \frac{\cost{S^\prime}{c}}{\epsilon}$.
    
    In the set of candidates, that \cref{alg:1l_median_3_candidates} returns, a curve $c^{\prime\prime}$ with up to $2\ell-2$ vertices from the union of the grid covers and distance at most $\frac{\epsilon\frac{2\delta n}{4\lvert S^\prime \rvert}\cost{S^\prime}{c}}{n} \leq \frac{\epsilon\frac{\delta\lvert T^\prime \rvert}{4\lvert S^\prime \rvert}\cost{S^\prime}{c}}{\lvert T^\prime \rvert} \leq \epsilon \frac{\cost{T^\prime}{c^\ast}}{\lvert T^\prime \rvert}$ between every corresponding pair of vertices of $c^\prime$ and $c^{\prime\prime}$ is contained. We conclude that $d_F(c^\prime, c^{\prime\prime}) \leq \frac{\epsilon \cost{T^\prime}{c^\ast}}{\lvert T^\prime \rvert}$.
    
    We can now bound the cost of $c^{\prime\prime}$ as follows:
    \begin{align*}
        \cost{T^\prime}{c^{\prime\prime}} & ={} \sum_{\tau \in T^\prime} d_F(\tau, c^{\prime\prime}) \leq \sum_{\tau \in T^\prime} \left(d_F(\tau, c^\prime) + \frac{\epsilon \cost{T^\prime}{c^\ast}}{\lvert T^\prime \rvert}\right) \\
        & \leq{} \sum_{\tau \in T^\prime} (d_F(\tau, c^\ast) + d_F(c^\ast, c^\prime)) + \epsilon \cost{T}{c^\ast} \\
        & \leq{} \sum_{\tau \in T^\prime} (d_F(\tau, c^\ast) + d_F(c^\ast, s) + d_F(s, c^\prime)) + \epsilon \cost{T^\prime}{c^\ast} \leq{} (3+3\epsilon) \cost{T^\prime}{c^\ast}.
    \end{align*}
    
    \textbf{Case 2:} $d_F(c,c^\ast) < (1+2\epsilon) \frac{\cost{T^\prime}{c^\ast}}{\lvert T^\prime \rvert}$
    
    The cost of $c$ can easily be bounded:
    \begin{align*}
        \cost{T^\prime}{c} \leq \sum_{\tau \in T^\prime} (d_F(\tau, c^\ast) + d_F(c^\ast, c)) < \cost{T^\prime}{c^\ast} + (1+2\epsilon) \cost{T^\prime}{c^\ast} = (2+2\epsilon) \cost{T^\prime}{c^\ast}.
    \end{align*}
    
    The claim follows by rescaling $\epsilon$ by $\frac{1}{3}$.
\end{proof}

Next we analyse the worst-case running-time of \cref{alg:1l_median_3_candidates} and the number of candidates it returns.

\begin{theorem}
    \label{theo:1l_median_3_running_time}
    \cref{alg:1l_median_3_candidates} has running-time and returns number of candidates $2^{O\left(\frac{(-\ln(\delta))^2 \cdot \beta}{\epsilon^2} + \log(m) \right)}$.
\end{theorem}
\begin{proof}
    The sample $S$ has size $O\left(\frac{-\ln(\delta) \cdot \beta}{\epsilon}\right)$ and sampling it takes time $O\left(\frac{-\ln(\delta) \cdot \beta}{\epsilon}\right)$. Let $n_S = \lvert S \rvert$. The for-loop runs \[\binom{n_S}{\frac{n_S}{2 \beta}} \in 2^{O\left(\frac{n_S}{2\beta} \log n_S\right)} \subset 2^{O\left(\frac{(-\ln(\delta))^2 \cdot \beta}{\epsilon^2}\right)}\] times. In each iteration, we run \cref{alg:1l_median_6}, taking time $O(m^2 \log(m) (-\ln^2 \delta) + m^3 \log m)$ (\conferre \cref{coro:1l_median_34}), we compute the cost of the returned curve with respect to $S^\prime$, taking time $O\left(\frac{-\ln(\delta)}{\epsilon} \cdot m \log(m)\right)$, and per curve in $S^\prime$ we build up to $m$ grids of size \[\left(\frac{\frac{(1+\epsilon)\Delta}{\epsilon}}{\frac{2\epsilon 2 \delta n \Delta}{n\sqrt{d} 4 \lvert S \rvert}}\right)^d = \left(\frac{\sqrt{d} \lvert S \rvert (1+\epsilon)}{\epsilon^2 \delta}\right)^d \in O\left(\frac{\beta^d(-\ln \delta)^d}{\epsilon^{3d}\delta^d}\right)\] each. For each curve $s \in S^\prime$, \cref{alg:1l_median_3_candidates} then enumerates all combinations of $2\ell-2$ points from these up to $m$ grids, resulting in \[O\left(\frac{m^{2\ell-2} \beta^{2\ell d-2d} (-\ln \delta)^{2\ell d-2d}}{\epsilon^{6\ell d-6d}\delta^{2\ell d-2d}}\right)\] candidates per $s \in S^\prime$, per iteration of the for-loop. Thus, \cref{alg:1l_median_3_candidates} computes $O\left(\poly{m,\beta,\delta^{-1},\epsilon^{-1}}\right)$ candidates per iteration of the for-loop and enumeration also takes time $O\left(\poly{m,\beta,\delta^{-1},\epsilon^{-1}}\right)$ per iteration of the for-loop.
    
    All in all, we have running-time and number of candidates $2^{O\left(\frac{(-\ln(\delta))^2 \cdot \beta}{\epsilon^2} + \log(m)\right)}$.
\end{proof}

\section{More Practical Approximation for \texorpdfstring{$(1,\ell)$-}{(1,l)-}Median by Simple Shortcutting}

The following algorithm is a modification of \cref{alg:1l_median_3_candidates}. It is more practical since it needs to cover only up to $m$ (small) balls, using grids. Unfortunately, it is not compatible with the superset-sampling technique and can therefore not be used as plugin in \cref{alg:kl_median}.

\begin{algorithm}[H]
\caption{$(1,\ell)$-Median by Simple Shortcutting\label{alg:1l_median_5}}
    \begin{algorithmic}[1]
        \Procedure{$(1,\ell)$-Median-$(5+\epsilon)$}{$T = \{ \tau_1, \dots, \tau_n \}, \delta, \epsilon$}
            \State $\widehat{c} \gets$ $(1,\ell)$-\textsc{Median}-$34$-\textsc{Approximation}$(T,\delta/2)$ \Comment{\cref{alg:1l_median_6}}
            \State $\Delta \gets \frac{\cost{T}{\widehat{c}}}{34}$, $\epsilon^\prime \gets \epsilon/9$, $P \gets \emptyset$
            \State $S \gets $ sample $\left\lceil -2(\epsilon^\prime)^{-1} (\ln(\delta)-\ln(4)) \right\rceil$ curves from $T$ uniformly and independently
            
            \hspace{\algorithmicindent} with replacement
            \State $W \gets$ sample $\lceil -64(\epsilon^\prime)^{-2}(\ln(\delta) - \ln(\lceil -8 (\epsilon^\prime)^{-1} (\ln(\delta)-\ln(4))\rceil)) \rceil$ curves from $T$ 
            
            \hspace{\algorithmicindent} uniformly and independently with replacement
            \State $c \gets \argmin\limits_{s \in S} \cost{W}{s}$
            \For{$i \in \{1, \dots, \lvert c \rvert\}$}
                \State $P \gets P \cup \mathbb{G}\left(B\left(v^c_i, \frac{(3+4\epsilon^\prime)}{n}34 \Delta\right),\frac{2 \epsilon^\prime \Delta}{n\sqrt{d}}\right)$ \Comment{$v^c_i$ is the $i$\textsuperscript{th} vertex of $c$}
            \EndFor
            \State $C \gets$ set of all polygonal curves with $2\ell-2$ vertices from $P$
            \State \Return $\argmin\limits_{c^\prime \in C} \cost{T}{c^\prime}$
        \EndProcedure
    \end{algorithmic}
\end{algorithm}

We prove the quality of approximation of \cref{alg:1l_median_5}.

\begin{theorem}
    Given two parameters $\delta, \epsilon \in (0,1)$ and a set $T = \{\tau_1, \dots, \tau_n \} \subset \mathbb{X}^d_m$ of polygonal curves, with probability at least $1-\delta$ \cref{alg:1l_median_5} returns a $(5+\epsilon)$-approximate $(1,\ell)$-median for $T$ with up to $2\ell-2$ vertices.
\end{theorem}
\begin{proof}
    Let $c^\ast \in \argmin\limits_{c \in \mathbb{X}^d_\ell} \cost{T}{c}$ be an optimal $(1,\ell)$-median for $T$.
 
    The expected distance between $s \in S$ and $c^\ast$ is \[\expected{d_F(s, c^\ast)} = \sum_{i=1}^n d_F(c^\ast, \tau_i) \cdot \frac{1}{n} = \frac{\cost{T}{c^\ast}}{n}.\]
 
    Now using Markov's inequality, for every $s \in S$ we have \[\Pr[d_F(s, c^\ast) > (1+\epsilon)\cost{T}{c^\ast}/n] \leq \frac{\cost{T}{c^\ast}n^{-1}}{(1+\epsilon)\cost{T}{c^\ast}n^{-1}} = \frac{1}{1+\epsilon},\] therefore by independence \[\Pr\left[\bigwedge_{s \in S} (d_F(s, c^\ast) > (1+\epsilon)\cost{T}{c^\ast}/n)\right] \leq \frac{1}{(1+\epsilon)^{\lvert S \rvert}} \leq \exp\left(-\frac{\epsilon \lvert S \rvert}{2}\right).\] Hence, with probability at most $\exp\left(-\frac{\epsilon \left\lceil -\frac{2(\ln(\delta) - \ln(4))}{\epsilon} \right\rceil}{2}\right) \leq \delta/4$ there is no $s \in S$ with $d_F(s, c^\ast) \leq (1+\epsilon) \frac{\cost{T}{c^\ast}}{n}$. Now, assume there is a $s \in S$ with $d_F(s, c^\ast) \leq (1+\epsilon) \cost{T}{c^\ast}/n$. We do not want any $t \in S \setminus \{s\}$ with $d_F(t, c^\ast) > (1+\epsilon)d_F(s, c^\ast)$ to have $\cost{W}{t} \leq \cost{W}{s}$. Using \cref{theo:indyk_median}, we conclude that this happens with probability at most \[\exp\left(-\frac{\epsilon^2 \lceil -64\epsilon^{-2}(\ln(\delta) - \ln(\lceil -8 (\epsilon^\prime)^{-1} (\ln(\delta)-\ln(4))\rceil)) \rceil}{64}\right) \leq \frac{\delta}{\lceil -8 (\epsilon^\prime)^{-1} (\ln(\delta)-\ln(4))\rceil} \leq \frac{\delta}{4 \lvert S \rvert},\] for each $t \in S \setminus \{s\}$. Also, with probability at most $\delta/2$ \cref{alg:1l_median_6} fails to compute a $34$-approximate $(1,\ell)$-median $\widehat{c} \in \mathbb{X}^d_\ell$ for $T$, \conferre \cref{coro:1l_median_34}.
    
    Using a union bound over these bad events, we conclude that with probability at least $1-\delta$, \cref{alg:1l_median_5} samples a curve $t \in S$ with $\cost{T}{t} \leq (1+\epsilon) \cost{T}{s}$ and \cref{alg:1l_median_6} computes a $34$-approximate $(1,\ell)$-median $\widehat{c} \in \mathbb{X}^d_\ell$ for $T$, i.e., $\cost{T}{c^\ast} \leq 34 \Delta = \cost{T}{\widehat{c}} \leq 34 \cost{T}{c^\ast}$. Let $v^{t}_1, \dots, v^{t}_{\lvert t \rvert}$ be the vertices of $t$. By \cref{lem:shortcutting_single} there exists a polygonal curve $c^\prime$ with up to $2\ell - 2$ vertices, every vertex contained in one of $B(v^{t}_1, d_F(c^\ast, t)), \dots, B(v^{t}_{\lvert t \rvert}, d_F(c^\ast, t))$ and $d_F(t, c^\prime) \leq d_F(t, c^\ast)$. Using the triangle-inequality yields
    \begin{align*}
        \sum_{\tau \in T} (d_F(t, c^\ast) - d_F(\tau, c^\ast)) \leq \sum_{\tau \in T} d_F(t, \tau) \leq (1+\epsilon) \sum_{\tau \in T} d_F(s, \tau) \leq (1+\epsilon) \sum_{\tau \in T} (d_F(\tau, c^\ast) + d_F(c^\ast, s)),
    \end{align*}
    which is equivalent to
    \begin{align*}
        n \cdot d_F(t, c^\ast) \leq (2+\epsilon) \cost{T}{c^\ast}  + (1+\epsilon) n  (1+\epsilon) \cost{T}{c^\ast}/n \Leftrightarrow{} & d_F(t, c^\ast) \leq (3 + 4\epsilon) \cost{T}{c^\ast}/n.
    \end{align*}
    
    Hence, we have $d_F(t, c^\prime) \leq d_F(t, c^\ast) \leq (3+4\epsilon) \cost{T}{c^\ast}/n \leq (3+ 4\epsilon) 34 \Delta/n$.
    
    In the last step, \cref{alg:1l_median_5} returns a curve $c^{\prime\prime}$ from the set $C$ of all curves with up to $2\ell-2$ vertices from $P$, the union of the grid covers, that evaluates best. We can assume that $c^{\prime\prime}$ has distance at most $\frac{\epsilon\Delta}{n} \leq \epsilon\frac{\cost{T}{c^\ast}}{n}$ between every corresponding pair of vertices of $c^\prime$ and $c^{\prime\prime}$. We conclude that $d_F(c^\prime, c^{\prime\prime}) \leq \frac{\epsilon\Delta}{n} \leq \epsilon\frac{\cost{T}{c^\ast}}{n}$.
    
    We can now bound the cost of $c^{\prime\prime}$ as follows:
    \begin{align*}
        \cost{T}{c^{\prime\prime}} & ={} \sum_{\tau \in T} d_F(\tau, c^{\prime\prime}) \leq \sum_{\tau \in T} \left(d_F(\tau, c^\prime) + \frac{\epsilon \Delta}{n}\right) \leq \sum_{\tau \in T} (d_F(\tau, t) + d_F(t, c^\prime)) + \epsilon \cost{T}{c^\ast} \\
        & \leq{} (1+\epsilon) \cost{T}{s} + (3+5\epsilon) \cost{T}{c^\ast} \\
        & \leq{} (1+\epsilon) \sum_{\tau \in T}(d_F(\tau, c^\ast) + d_F(c^\ast, s)) + (3+5\epsilon) \cost{T}{c^\ast} \\
        & \leq{} (1+\epsilon)\cost{T}{c^\ast} + (1+\epsilon)^2 \cost{T}{c^\ast} + (3+5\epsilon) \cost{T}{c^\ast} \\
        & \leq{} (5 + 9\epsilon) \cost{T}{c^\ast}
    \end{align*}
    
    The claim follows by rescaling $\epsilon$ by $\frac{1}{9}$.
\end{proof}

We analyse the worst-case running-time of \cref{alg:1l_median_5}.

\begin{theorem}
    \cref{alg:1l_median_5} has running-time $O\left(\frac{nm^{2\ell-1} \log(m)}{\epsilon^{(2\ell-2)d}} + \frac{m^3 \log(m) (-\ln(\delta))^2}{\epsilon^3}\right)$.
\end{theorem}
\begin{proof}
    \cref{alg:1l_median_6} has running-time $O(m^2 \log(m) (-\ln^2 \delta)) + m^3 \log m)$. The sample $S$ has size $O\left(\frac{-\ln(\delta)}{\epsilon}\right)$ and the sample $W$ has size $O\left(\frac{-\ln(\delta)}{\epsilon^2}\right)$. Evaluating each curve of $S$ against $W$ takes time $O\left(\frac{m^2 \log(m) (-\ln(\delta))^2}{\epsilon^3}\right)$, using the algorithm of \citet{alt_godau} to compute the distances.
    
    Now, $c$ has up to $m$ vertices and every grid consists of $\left( \frac{\frac{(3+\epsilon)\Delta}{n}}{\frac{2\epsilon^\prime\Delta}{nc\sqrt{d}}} \right)^d = \left( \frac{(3+\epsilon)c\sqrt{d}}{2 \epsilon^\prime} \right)^d \in O\left(\frac{1}{\epsilon^d}\right)$ points. Therefore, we have $O\left(\frac{m}{\epsilon^{d}}\right)$ points in $P$ and \cref{alg:1l_median_5} enumerates all combinations of $2\ell-2$ points from $P$ taking time $O\left(\frac{m^{2\ell-2}}{\epsilon^{(2\ell-2)d}}\right)$. Afterwards, these candidates are evaluated, which takes time $O(n m \log(m))$ per candidate using the algorithm of \citet{alt_godau} to compute the distances. All in all, we then have  running-time $O\left(\frac{nm^{2\ell-1} \log(m)}{\epsilon^{(2\ell-2)d}} + \frac{m^3 \log(m) (-\ln(\delta))^2}{\epsilon^3}\right)$.
\end{proof}

\section{\texorpdfstring{$(1+\epsilon)$-Approximation}{Approximation-Scheme} for \texorpdfstring{$(1,\ell)$-}{(1,l)-}Median by Advanced Shortcutting}
\FloatBarrier

Now we present \cref{alg:1l_median_1_candidates}, which returns candidates, containing a $(1+\epsilon)$-approximate $(1,\ell)$-median of complexity at most $2\ell-2$, for a cluster contained in the input, that takes a constant fraction of the input, w.h.p. Before we present the algorithm, we present our second shortcutting lemma. Here, we do not introduce shortcuts with respect to a single curve, but with respect to several curves: by introducing shortcuts with respect to $\epsilon \lvert T \rvert$ well-chosen curves from the given set $T \subset \mathbb{X}^d_m$ of polygonal curves, for a given $\epsilon \in (0,1)$, we preserve the distances to at least $(1-\epsilon)\lvert T \rvert$ curves from $T$. In this context well-chosen means that there exists a certain number of subsets of $T$, of each we have to pick a curve for shortcutting. This will enable the high quality of approximation of \cref{alg:1l_median_1_candidates} and we formalize this in the following lemma. 

\begin{figure}
    \centering
    \includegraphics[width=0.9\textwidth]{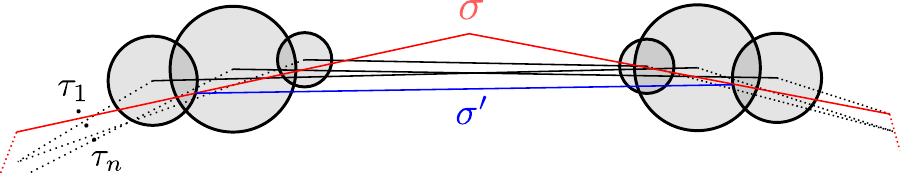}
    \caption{Visualization of an advanced shortcut. The black curves are input-curves and the red curve is an optimal median. By inserting the blue shortcut we can find a curve that has distance not larger as the median to all the black curves, but one, and with all vertices contained in the balls centered at the black curves vertices.}
    \label{fig:advanced_shortcutting}
\end{figure}

\begin{lemma}[shortcutting using a set of polygonal curves]
    \label{lem:shortcutting_multiple}
    Let $\sigma \in \mathbb{X}^d$ be a polygonal curve with $\lvert \sigma \rvert > 2$ vertices and $T = \{ \tau_1, \dots, \tau_n \} \subset \mathbb{X}^d$ be a set of polygonal curves. For $i \in \{1, \dots, n\}$, let $r_i = d_F(\tau_i, \sigma)$ and for $j \in \{1, \dots, \lvert \tau_i \rvert \}$, let $v^{\tau_i}_j$ be the $j$\textsuperscript{th} vertex of $\tau_i$. 
    
    For any $\epsilon \in (0,1)$ there are $2\lvert \sigma \rvert-4$ subsets $T_1, \dots, T_{2\lvert \sigma \rvert-4} \subseteq T$ of $\frac{\epsilon n}{2\lvert \sigma \rvert}$ curves each (not necessarily disjoint) such that for every subset $T^\prime \subseteq T$ containing at least one curve out of each $T_k \in \{T_1, \dots, T_{2\lvert \sigma \rvert-4}\}$, a polygonal curve $\sigma^\prime \in \mathbb{X}^d$ exists with every vertex contained in \[\bigcup\limits_{\tau_i \in T^\prime} \bigcup\limits_{j \in \{1, \dots, \lvert \tau_i \rvert \}} B(v^{\tau_i}_j, r_i),\] $d_F(\tau, \sigma^\prime) \leq d_F(\tau, \sigma)$ for each $\tau \in T\setminus (T_1 \cup \dots \cup T_{2 \lvert \sigma \rvert-4})$ and $\lvert \sigma^\prime \rvert \leq 2 \lvert \sigma \rvert - 2$.
\end{lemma}
The idea is the following, see \cref{fig:advanced_shortcutting} for a visualization. One can argue that every vertex $v$ of $\sigma$ not contained in any of the balls centered at the vertices of the curves in $T$ (and of radius according to their distance to $\sigma$) can be shortcut by connecting the last point $p^{-}$ before $v$ (in terms of the parameter of $\sigma$) contained in one ball and first point $p^{+}$ after $v$ contained in one ball. This does not increase the Fr\'echet distances between $\sigma$ and the $\tau \in T$, because only matchings among line segments are affected by this modification. Furthermore, most distances are preserved when we do not actually use the last and first ball before and after $v$, but one of the $\frac{\epsilon n}{2\vert\sigma\rvert}$ balls before and one of the $\frac{\epsilon n}{2\lvert\sigma\rvert}$ balls after $v$, which is the key of the following proof.
\begin{proof}[Proof of \cref{lem:shortcutting_multiple}]
    For the sake of simplicity, we assume that $\frac{\epsilon n}{2\lvert \sigma\rvert}$ is integral. Let $\ell = \lvert \sigma \rvert$. For $i \in \{1, \dots, n\}$, let $v^{\tau_i}_1, \dots, v^{\tau_i}_{\lvert \tau_i \rvert}$ be the vertices of $\tau_i$ with instants $t^{\tau_i}_{1}, \dots, t^{\tau_i}_{\lvert \tau_i \rvert}$ and let $v^\sigma_1, \dots, v^\sigma_{\ell}$ be the vertices of $\sigma$ with instants $t^\sigma_1, \dots, t^\sigma_{\ell}$. Also, for $h \in \mathcal{H}$ (recall that $\mathcal{H}$ is the set of all continuous bijections $h\colon [0,1] \rightarrow [0,1]$ with $h(0) = 0$ and $h(1) = 1$) and $i \in \{1, \dots, n\}$, let $r_{i,h} = \max\limits_{t \in [0,1]} \lVert \sigma(t) - \tau_i(h(t)) \rVert$ be the distance realized by $h$ with respect to $\tau_i$. We know from \cref{lem:sequence} that for each $i \in \{1, \dots, n\}$ there exists a sequence $(h_{i,x})_{x=1}^\infty$ in $\mathcal{H}$, such that $\lim\limits_{x \to \infty} r_{i,h_{i,x}} = d_F(\sigma, \tau_i) = r_i$.
    \\
    
    In the following, given arbitrary $h_1, \dots, h_n \in \mathcal{H}$, we describe how to modify $\sigma$, such that its vertices can be found in the balls around the vertices of the $\tau \in T$, of radii determined by $h_1, \dots, h_n$. Later we will argue that this modification can be applied using the $h_{1,x}, \dots, h_{n,x}$, for each $x \in \mathbb{N}$, in particular.
    \\
    
    Now, fix arbitrary $h_1, \dots, h_n \in \mathcal{H}$ and for an arbitrary $k \in \{2, \dots, \lvert \sigma \rvert-1\}$, fix the vertex $v^\sigma_k$ of $\sigma$ with instant $t^\sigma_k$. For $i \in \{1, \dots, n\}$, let $s_i$ be the maximum of $\{1, \dots, \lvert \tau_i \rvert-1\}$, such that $t^{\tau_i}_{s_i} \leq h_i(t^\sigma_k) \leq t^{\tau_i}_{s_{i}+1}$. Namely, $v^\sigma_k$ is matched to a point on the line segment $\overline{v^{\tau_1}_{s_1}v^{\tau_1}_{s_1+1}}, \dots, \overline{v^{\tau_n}_{s_n}v^{\tau_n}_{s_n+1}}$, respectively, by $h_1, \dots, h_n$.
    
    For $i \in \{1, \dots, n\}$, let $t^{-}_i$ be the maximum of $[0, t^\sigma_k]$, such that $\sigma(t^{-}_i) \in B(v^{\tau_i}_{s_i}, r_{i,h_i})$ and let $t^{+}_i$ be the minimum of $[t^\sigma_k, 1]$, such that $\sigma(t^+_i) \in B(v^{\tau_i}_{s_i+1}, r_{i,h_i})$. These are the instants when $\sigma$ visits $B(v^{\tau_i}_{s_i}, r_{i,h_i})$ before or when visiting $v^\sigma_k$ and $\sigma$ visits $B(v^{\tau_i}_{s_i+1}, r_{i,h_i})$ when or after visiting $v^\sigma_k$, respectively. Furthermore, there is a permutation $\alpha \in S_n$ of the index set $\{1, \dots, n\}$, such that 
    \begin{equation}
        t^{-}_{\alpha^{-1}(1)} \leq \dots \leq t^{-}_{\alpha^{-1}(n)}. \label{eq:tminus} \tag{I}
    \end{equation} 
    Also, there is a permutation $\beta \in S_n$ of the index set $\{1, \dots, n\}$, such that 
    \begin{equation}
        t^{+}_{\beta^{-1}(1)} \leq \dots \leq t^{+}_{\beta^{-1}(n)}. \label{eq:tplus} \tag{II}
    \end{equation}
    Additionally, for each $i \in \{1, \dots, n\}$ we have 
    \begin{equation}
        t^{\tau_i}_{s_i} \leq h_i(t^{-}_i) \label{eq:tsitminus} \tag{III}
    \end{equation}
    and
    \begin{equation}
        h_i(t^{+}_i) \leq t^{\tau_i}_{s_i+1}, \label{eq:tsitplus} \tag{IV}
    \end{equation}
    because $\sigma(t^-_i)$ and $\sigma(t^+_i)$ are closest points to $v^\sigma$ on $\sigma$ that have distance at most $r_{i,h_i}$ to $v^{\tau_i}_{s_i}$ and $v^{\tau_i}_{s_i+1}$, respectively, by definition. We will now use \cref{eq:tminus,eq:tplus,eq:tsitminus,eq:tsitplus} to prove that an advanced shortcut only affects matchings among line segments and hence we can easily bound the resulting distances for at least $(1-\epsilon)n$ of the curves.
    
    Let \[I_{v^\sigma_k}(h_1, \dots, h_n) = \{ \tau_{\alpha^{-1}((1-\frac{\epsilon}{2\ell})n + 1)}, \dots, \tau_{\alpha^{-1}(n)}\},\ O_{v^\sigma_k}(h_1, \dots, h_n) = \{ \tau_{\beta^{-1}(1)}, \dots, \tau_{\beta^{-1}(\frac{\epsilon n}{2\ell})} \}.\] $I_{v^\sigma_k}(h_1, \dots, h_n)$ is the set of the last $\frac{\epsilon n}{2\ell}$ curves whose balls are visited by $\sigma$, before or when $\sigma$ visits $v^\sigma_k$. Similarly, $O_{v^\sigma_k}(h_1, \dots, h_n)$ is the set of the first $\frac{\epsilon n}{2\ell}$ curves whose balls are visited by $\sigma$, when or immediately after $\sigma$ visited $v^\sigma_k$. We now modify $\sigma$, such that $v^\sigma_k$ is replaced by two new vertices that are elements of at least one $B(v^{\tau_i}_{j}, r_{i,h_i})$, for a $\tau_i \in I_{v^\sigma_k}(h_1, \dots, h_n)$, respectively for a $\tau_i \in O_{v^\sigma_k}(h_1, \dots, h_n)$, and $j \in \{1, \dots, \lvert \tau_i \rvert \}$, each.
    
    Let $\sigma^\prime_{h_1, \dots, h_n}$ be the piecewise defined curve, defined just like $\sigma$ on $\left[0,t^-_{\alpha^{-1}(k_1)}\right]$ and $\left[t^+_{\beta^{-1}(k_2)},1\right]$ for arbitrary $k_1 \in \{(1-\frac{\epsilon}{2\ell})n + 1, \dots, n\}$ and $k_2 \in \{1, \dots, \frac{\epsilon n}{2\ell} \}$, but on $\left(t^-_{\alpha^{-1}(k_1)}, t^+_{\beta^{-1}(k_2)}\right)$ it connects $\sigma\left(t^-_{\alpha^{-1}(k_1)}\right)$ and $\sigma\left(t^+_{\beta^{-1}(k_2)}\right)$ with the line segment \[\gamma(t) = \left(1-\frac{t-t^-_{\alpha^{-1}(k_1)}}{t^+_{\beta^{-1}(k_2)}-t^-_{\alpha^{-1}(k_1)}}\right) \sigma\left(t^-_{\alpha^{-1}(k_1)}\right) + \frac{t-t^-_{\alpha^{-1}(k_1)}}{t^+_{\beta^{-1}(k_2)}-t^-_{\alpha^{-1}(k_1)}} \sigma\left(t^+_{\beta^{-1}(k_2)}\right).\] We now argue that for all $\tau_i \in T \setminus (I_{v^\sigma_k}(h_1, \dots, h_n) \cup O_{v^\sigma_k}(h_1, \dots, h_n))$ the Fr\'echet distance between $\sigma^\prime_{h_1, \dots, h_n}$ and $\tau_i$ is upper bounded by $r_{i,h_i}$. First, note that by definition $h_1, \dots, h_n$ are strictly increasing functions, since they are continuous bijections that map $0$ to $0$ and $1$ to $1$. As immediate consequence, we have that 
    \begin{equation}
        t^{\tau_i}_{s_i} \leq h_i(t^-_i) \leq h_i\left(t^-_{\alpha^{-1}(k_1)}\right) \label{eq:afterin} \tag{V}
    \end{equation}
    for each $\tau_i \in T \setminus I_{v^\sigma_k}(h_1, \dots, h_n)$ and
    \begin{equation}
        h_i\left(t^+_{\beta^{-1}(k_2)}\right) \leq h_i(t^+_i) \leq t^{\tau_i}_{s_i+1} \label{eq:beforeout} \tag{VI}
    \end{equation}
    for each $\tau_i \in T \setminus O_{v^\sigma_k}(h_1, \dots, h_n)$, using \cref{eq:tminus,eq:tplus,eq:tsitminus,eq:tsitplus}. Therefore, each $\tau_i \in T \setminus (I_{v^\sigma_k}(h_1, \dots, h_n) \cup O_{v^\sigma_k}(h_1, \dots, h_n))$ has no vertex between the instants $h_i\left(t^-_{\alpha^{-1}(k_1)}\right)$ and $h_i\left(t^+_{\beta^{-1}(k_2)}\right)$. We also know that for each $\tau_i \in T$
    \begin{equation}
        \left\lVert \sigma\left(t^-_{\alpha^{-1}(k_1)}\right) - \tau_i\left(h_i\left(t^-_{\alpha^{-1}(k_1)}\right)\right) \right\rVert \leq r_{i,h_i} \label{eq:diststart} \tag{VII}
    \end{equation}
    and 
    \begin{equation}
        \left\lVert \sigma\left(t^+_{\beta^{-1}(k_2)}\right) - \tau_i\left(h_i\left(t^+_{\beta^{-1}(k_2)}\right)\right) \right\rVert \leq r_{i,h_i}. \label{eq:distend} \tag{VIII}
    \end{equation}
    
    Let $D_{s,\sigma} = \left[0,t^-_{\alpha^{-1}(k_1)}\right)$, $D_{m,\sigma} = \left[t^-_{\alpha^{-1}(k_1)}, t^+_{\beta^{-1}(k_2)}\right]$ and $D_{e,\sigma} = \left(t^+_{\beta^{-1}(k_2)}, 1\right]$. Also, for $i \in \{1, \dots, n\}$, let $D_{s,\tau_i} = \left[0,h_i\left(t^-_{\alpha^{-1}(k_1)}\right)\right)$, $D_{m,\tau_i} = \left[h_i\left(t^-_{\alpha^{-1}(k_1)}\right), h_i\left(t^+_{\beta^{-1}(k_2)}\right)\right]$ and $D_{e,\tau_i} = \left(h_i\left(t^+_{\beta^{-1}(k_2)}\right),1\right]$. Now, for each $\tau_i \in T \setminus (I_{v^\sigma_k}(h_1, \dots, h_n) \cup O_{v^\sigma_k}(h_1, \dots, h_n))$ we use $h_i$ to match $\sigma^\prime_{h_1, \dots, h_n}\vert_{D_{s,\sigma}}$ to $\tau_i\vert_{D_{s,\tau_i}}$ and $\sigma^\prime_{h_1, \dots, h_n}\vert_{D_{e,\sigma}}$ to $\tau_i\vert_{D_{e,\tau_i}}$ with distance at most $r_{i,h_i}$. Since $\sigma^\prime_{h_1, \dots, h_n}\vert_{D_{m,\sigma}}$ and $\tau_i\vert_{D_{m,\tau_i}}$ are just line segments by \cref{eq:afterin,eq:beforeout}, they can be matched to each other with distance at most \[\max\left\{\left\lVert \sigma\left(t^-_{\alpha^{-1}(k_1)}\right) - \tau_i\left(h_i\left(t^-_{\alpha^{-1}(k_1)}\right)\right) \right\rVert, \left\lVert \sigma\left(t^+_{\beta^{-1}(k_2)}\right) - \tau_i\left(h_i\left(t^+_{\beta^{-1}(k_2)}\right)\right) \right\rVert \right\},\] which is at most $r_{i,h_i}$ by \cref{eq:diststart,eq:distend}. We conclude that $d_F(\sigma^\prime_{h_1, \dots, h_n}, \tau_i) \leq r_{i,h_i}$. 
    
    Because this modification works for every $h_1, \dots, h_n \in \mathcal{H}$, we conclude that $d_F(\sigma^\prime_{h_1, \dots, h_n}, \tau_i) \leq r_{i,h_i}$ for every $h_1, \dots, h_n \in \mathcal{H}$ and $\tau_i \in T \setminus (I_{v^\sigma_k}(h_1, \dots, h_n) \cup O_{v^\sigma_k}(h_1, \dots, h_n))$. Thus $\lim\limits_{x \to \infty} d_F(\sigma^\prime_{h_{1,x}, \dots, h_{n,x}}, \tau_i) \leq d_F(\sigma, \tau_i) = r_i$ for each $\tau_i \in T \setminus (I_{v^\sigma_k}(h_{1,x}, \dots, h_{n,x}) \cup O_{v^\sigma_k}(h_{1,x}, \dots, h_{n,x}))$.
    \\
    
    Now, to prove the claim, for each combination $h_1, \dots, h_n \in \mathcal{H}$, we apply this modification to $v^\sigma_k$ and successively to every other vertex $v^{\sigma^\prime_{h_1, \dots, h_n}}_l$ of the resulting curve $\sigma^\prime_{h_1, \dots, h_n}$, except $v^{\sigma^\prime_{h_1, \dots, h_n}}_1$ and $v^{\sigma^\prime_{h_1, \dots, h_n}}_{\lvert\sigma^\prime_{h_1, \dots, h_n}\rvert}$, since these must be elements of $B(v^{\tau_i}_1, r_{i,h_i})$ and $B(v^{\tau_i}_{\lvert \tau_i \rvert}, r_{i,h_i})$, respectively, for each $i \in \{1, \dots, n\}$, by definition of the Fr\'echet distance. 
    
    Since the modification is repeated at most $\lvert \sigma \rvert - 2$ times for each combination $h_1, \dots h_n \in \mathcal{H}$, we conclude that the number of vertices of each $\sigma^\prime_{h_1, \dots, h_n}$ can be bounded by $2 \cdot (\lvert \sigma \rvert - 2) + 2$. 
    \\
   
    $T_1, \dots, T_{2\ell-4}$ are therefore all the $I_{v^\sigma_k}(h_{1,x}, \dots, h_{n,x})$ and $O_{v^\sigma_k}(h_{1,x}, \dots, h_{n,x})$ for $k \in \{2, \dots, 2\lvert \sigma\rvert - 3\}$, when $x \to \infty$. Note that every $I_{v^\sigma_k}(h_{1,x}, \dots, h_{n,x})$ and $O_{v^\sigma_k}(h_{1,x}, \dots, h_{n,x})$ is determined by the visiting order of the balls and since their radii converge, these sets do too.
\end{proof}

We now present \cref{alg:1l_median_1_candidates}, which is nearly identical to \cref{alg:1l_median_3_candidates} but uses the advanced shortcutting lemma. Furthermore, like \cref{alg:1l_median_3_candidates}, it can be used as plugin in the recursive $k$-median approximation-scheme (\cref{alg:kl_median}) that we present in \cref{sec:k-median}.

\begin{algorithm}[H]
\caption{$(1,\ell)$-Median for Subset by Advanced Shortcutting\label{alg:1l_median_1_candidates}}
    \begin{algorithmic}[1]
        \Procedure{$(1,\ell)$-Median-$(1+\epsilon)$-Candidates}{$T = \{ \tau_1, \dots, \tau_n \}, \beta, \delta, \epsilon$}
            \State $\epsilon^\prime \gets \epsilon/6$, $C \gets \emptyset$
            \State $S \gets $ sample $\left\lceil -8\beta \ell (\epsilon^\prime)^{-1} (\ln(\delta)-\ln(4(2\ell-4))) \right\rceil$ curves from $T$ uniformly and independently
            
            \hspace{\algorithmicindent} with replacement
            \For{$S^\prime \subseteq S$ with $\lvert S^\prime \rvert = \frac{\lvert S \rvert}{2 \beta}$}
                \State $c \gets$ $(1,\ell)$-\textsc{Median}-$34$-\textsc{Approximation}$(S^\prime,\delta/4)$ \Comment{\cref{alg:1l_median_6}}
                \State $\Delta \gets \cost{S^\prime}{c}$, $\Delta_l \gets \frac{2\delta n}{4 \lvert S \rvert} \frac{\Delta}{34}$, $\Delta_u \gets \frac{1}{\epsilon^\prime} \Delta$
                \State $C \gets C \cup \{c\}$, $P \gets \emptyset$
                \For{$s \in S^\prime$}
                    \For{$i \in \{1, \dots, \lvert s \rvert\}$}
                        \State $P \gets P \cup \mathbb{G}\left(B\left(v^s_i, \frac{4\ell}{\epsilon^\prime}\Delta_u\right),\frac{2\epsilon^\prime}{n\sqrt{d}} \Delta_l\right)$ \Comment{$v^s_i$: $i$\textsuperscript{th} vertex of $s$}
                    \EndFor
                \EndFor
                \State $C \gets C\ \cup$ set of all polygonal curves with $2\ell-2$ vertices from $P$
            \EndFor
            \State \Return $C$
        \EndProcedure
    \end{algorithmic}
\end{algorithm}

We prove the quality of approximation of \cref{alg:1l_median_1_candidates}.

\begin{theorem}
    \label{theo:1l_median_1_guarantee}
    Given three parameters $\beta \in [1, \infty)$, $\delta \in (0,1)$, $\epsilon \in (0,0.158]$ and a set $T = \{ \tau_1, \dots, \tau_n \} \subset \mathbb{X}^d_m$ of polygonal curves, with probability at least $1-\delta$ the set of candidates that \cref{alg:1l_median_1_candidates} returns contains a $(1+\epsilon)$-approximate $(1,\ell)$-median with up to $2\ell-2$ vertices for any $T^\prime \subseteq T$, if $\lvert T^\prime \rvert \geq \frac{1}{\beta} \lvert T \rvert$.
\end{theorem}
In the following proof we make use of a case distinction developed by \citet[Proof of Theorem 10]{abhin2020kmedian}, which is a key ingredient to enable the $(1+\epsilon)$-approximation, though the domain of $\epsilon$ has to be restricted to $(0, 0.158]$.
\begin{proof}[Proof of \cref{theo:1l_median_1_guarantee}]
    We assume that $\lvert T^\prime \rvert \geq \frac{1}{\beta} \lvert T \rvert$. Let $n^\prime$ be the number of sampled curves in $S$ that are elements of $T^\prime$. Clearly, $\expected{n^\prime} \geq \sum_{i=1}^{\lvert S \rvert} \frac{1}{\beta} = \frac{\lvert S \rvert}{\beta}$. Also $n^\prime$ is the sum of independent Bernoulli trials. A Chernoff bound (\conferre \cref{lem:bernoulli_trial_bound}) yields:
    \begin{align*}
        \Pr\left[n^\prime < \frac{\lvert S \rvert}{2\beta}\right] \leq \Pr\left[n^\prime < \frac{\expected{n^\prime}}{2}\right] \leq \exp\left(-\frac{1}{4}\frac{\lvert S \rvert}{2\beta}\right) \leq \exp\left( \frac{\ell(\ln(\delta)-\ln(4(2\ell-4)))}{\epsilon} \right) \leq \left(\frac{\delta^\ell}{4^\ell}\right)^{\frac{1}{\epsilon}} \leq \frac{\delta}{8}.
    \end{align*}
    
    In other words, with probability at most $\delta/8$ no subset $S^\prime \subseteq S$, of cardinality at least $\frac{\lvert S \rvert}{2\beta}$, is a subset of $T^\prime$. We condition the rest of the proof on the contrary event, denoted by $\mathcal{E}_{T^\prime}$, namely, that there is a subset $S^\prime \subseteq S$ with $S^\prime \subseteq T^\prime$ and $\lvert S^\prime \rvert \geq \frac{\lvert S \rvert}{2\beta}$. Note that $S^\prime$ is then a uniform and independent sample of $T^\prime$.

    Now, let $c^\ast \in \mathbb{X}^d_\ell$ be an optimal $(1,\ell)$-median for $T^\prime$. The expected distance between $s \in S^\prime$ and $c^\ast$ is \[\expected{d_F(s, c^\ast)\ \vert\ \mathcal{E}_{T^\prime}} = \sum_{\tau \in T^\prime} d_F(c^\ast, \tau) \cdot \frac{1}{\lvert T^\prime \rvert} =  \frac{\cost{T^\prime}{c^\ast}}{\lvert T^\prime \rvert}.\] By linearity we have $\expected{\cost{S^\prime}{c^\ast}\ \vert\ \mathcal{E}_{T^\prime}} = \frac{\lvert S^\prime \rvert}{\lvert T^\prime \rvert} \cost{T^\prime}{c^\ast}$. Markov's inequality yields:
    \begin{align*}
        \Pr\left[ \frac{\delta\lvert T^\prime \rvert}{4\lvert S^\prime \rvert}\cost{S^\prime}{c^\ast} >  \cost{T^\prime}{c^\ast}\ \Big\vert\ \mathcal{E}_{T^\prime}\right] \leq \frac{\delta}{4}.
    \end{align*}
    
    We conclude that with probability at most $\delta/4$ we have $\frac{\delta \lvert T^\prime \rvert}{4\lvert S^\prime \rvert} \cost{S^\prime}{c^\ast} > \cost{T^\prime}{c^\ast}$.
    
    Now, from \cref{lem:shortcutting_multiple} we know that there are $2\ell-4$ subsets $T^\prime_1, \dots, T^\prime_{2\ell-4} \subseteq T^\prime$, of cardinality $\frac{\epsilon \lvert T^\prime \rvert}{2\ell}$ each and which are not necessarily disjoint, such that for every set $W \subseteq T^\prime$ that contains at least one curve $\tau \in T^\prime_i$ for each $i \in \{1, \dots, 2\ell-4\}$, there exists a curve $c^\prime \in \mathbb{X}^d_{2\ell-2}$ which has all of its vertices contained in \[\bigcup\limits_{\tau \in W} \bigcup\limits_{j \in \{1, \dots, \lvert \tau \rvert\}} B(v^{\tau}_j, d_F(\tau, c^\ast))\] and for at least $(1-\epsilon) \lvert T^\prime \rvert$ curves $\tau \in T^\prime \setminus (T^\prime_1 \cup \dots \cup T^\prime_{2\ell-4})$ it holds that $d_F(\tau, c^\prime) \leq d_F(\tau, c^\ast)$.
    
    There are up to $\frac{\epsilon \lvert T^\prime \rvert}{4\ell}$ curves with distance to $c^\ast$ at least $\frac{4\ell \cost{T^\prime}{c^\ast}}{\epsilon \lvert T^\prime \rvert}$. Otherwise the cost of these curves would exceed $\cost{T^\prime}{c^\ast}$, which is a contradiction. Later we will prove that each ball we cover has radius at most $\frac{4\ell \cost{T^\prime}{c^\ast}}{\epsilon \lvert T^\prime \rvert}$. Therefore, for each $i \in \{1, \dots, 2\ell-4\}$ we have to ignore up to half of the curves $\tau \in T^\prime_i$, since we do not cover the balls of radius $d_F(\tau, c^\ast)$ centered at their vertices. For each $i \in \{1, \dots, 2\ell-4\}$ and $s \in S^\prime$ we now have \[\Pr\left[s \in T^\prime_i \wedge d_F(s, c^\ast) \leq \frac{4\ell\cost{T^\prime}{c^\ast}}{\epsilon \lvert T^\prime \rvert} \ \Big\vert\ \mathcal{E}_{T^\prime} \right] \geq \frac{\epsilon}{4\ell}.\] Therefore, by independence, for each $i \in \{1, \dots, 2\ell-4\}$ the probability that no $s \in S^\prime$ is an element of $T^\prime_i$ and has distance to $c^\ast$ at most $\frac{4\ell \cost{T^\prime}{c^\ast}}{\epsilon \lvert T^\prime \rvert}$ is at most $(1-\frac{\epsilon}{4\ell})^{\lvert S^\prime \rvert} \leq \exp\left(-\frac{\epsilon}{4\ell} \frac{4\ell(\ln(4(2\ell-4))-\ln(\delta))}{\epsilon}\right) = \exp\left(\ln\left(\frac{\delta}{4(2\ell-4)}\right)\right) = \frac{\delta}{4(2\ell-4)}$. Also, with probability at most $\delta/4$ \cref{alg:1l_median_6} fails to compute a $34$-approximate $(1,\ell)$-median $c \in \mathbb{X}^d_\ell$ for $S^\prime$, \conferre \cref{coro:1l_median_34}.
    
    Using a union bound over these bad events, we conclude that with probability at least $1-7/8 \delta$ all of the following events occur simultaneously:
    \begin{enumerate}
        \item \label{item1} There is a subset $S^\prime \subseteq S$ of cardinality at least $\lvert S \rvert /(2\beta)$ that is a uniform and independent sample of $T^\prime$,
        \item \label{item2} for each $i \in \{1, \dots, 2\ell-4\}$, $S^\prime$ contains at least one curve from $T^\prime_i$ with distance to $c^\ast$ up to $\frac{4\ell\cost{T^\prime}{c^\ast}}{\epsilon \lvert T^\prime \rvert}$,
        \item \label{item3} \cref{alg:1l_median_6} computes a polygonal curve $c \in \mathbb{X}^d_\ell$ with $\cost{S^\prime}{c^\ast_{S^\prime}} \leq \cost{S^\prime}{c} \leq 34 \cost{S^\prime}{c^\ast_{S^\prime}}$, where $c^\ast_{S^\prime} \in \mathbb{X}^d_\ell$ is an optimal $(1,\ell)$-median for $S^\prime$,
        \item \label{item4} and it holds that $\frac{\delta \lvert T^\prime \rvert}{4 \lvert S^\prime \rvert} \cost{S^\prime}{c^\ast} \leq \cost{T^\prime}{c^\ast}$.
    \end{enumerate}

    Let $B_{c^\ast} = \left\{ \tau \in T^\prime \mid d_F(\tau, c^\ast) \leq \frac{\cost{T^\prime}{c^\ast}}{\epsilon^2 \lvert T^\prime \rvert} \right\}$, $T^\prime_{c^\ast} = T^\prime \cap B_{c^\ast}$ and $B_{c} = \left\{ \tau \in T^\prime \mid d_F(\tau, c) \leq \epsilon \frac{\cost{T^\prime}{c^\ast}}{\lvert T^\prime \rvert} \right\}$. First, note that $\lvert T^{\prime} \setminus B_{c^\ast} \rvert \leq \epsilon^2 \lvert T^\prime \rvert$, otherwise $\cost{T^\prime \setminus B_{c^\ast}}{c^\ast} > \cost{T^\prime}{c^\ast}$, which is a contradiction, and therefore $\lvert T^\prime_{c^\ast} \rvert \geq (1-\epsilon^2) \lvert T^\prime \rvert$. We now distinguish two cases:\\
    
    \textbf{Case 1:} $\lvert T^\prime_{c^\ast} \setminus B_{c} \rvert > 2 \epsilon \lvert T^\prime_{c^\ast} \rvert$
    
    We have $2 \epsilon \lvert T^\prime_{c^\ast} \rvert \geq (1-\epsilon^2)2\epsilon \lvert T^\prime \rvert \geq \epsilon \lvert T^\prime \rvert$, hence $\Pr\left[d_F(s,c) > \epsilon \frac{\cost{T^\prime}{c^\ast}}{\lvert T^\prime \rvert}\ \Big\vert\ \mathcal{E}_{T^\prime} \right] \geq \epsilon$ for each $s \in S^\prime$. Using independence we conclude that with probability at most \[(1-\epsilon)^{\lvert S^\prime \rvert} \leq \exp\left(-\epsilon \frac{4\ell(\ln(4(2\ell-4))-\ln(\delta))}{\epsilon} \right) \leq \frac{\delta^{4\ell}}{4^{4\ell}} \leq \frac{\delta}{8}\] no $s \in S^\prime$ has distance to $c$ greater than $\epsilon\frac{\cost{T^\prime}{c^\ast}}{\lvert T^\prime \rvert}$. Using a union bound again, we conclude that with probability at least $1-\delta$ \cref{item1,item2,item3,item4} occur simultaneously and at least one $s \in S^\prime$ has distance to $c$ greater than $\epsilon\frac{\cost{T^\prime}{c^\ast}}{\lvert T^\prime \rvert}$, hence $\cost{S^\prime}{c} > \epsilon\frac{\cost{T^\prime}{c^\ast}}{\lvert T^\prime \rvert} \Leftrightarrow \frac{\cost{S^\prime}{c}}{\epsilon} > \frac{\cost{T^\prime}{c^\ast}}{\lvert T^\prime \rvert}$ and thus we indeed cover the balls of radius at most $\frac{4\ell \cost{T^\prime}{c^\ast}}{\epsilon \lvert T^\prime \rvert} < \frac{4\ell}{\epsilon} \frac{\cost{S^\prime}{c^\ast}}{\epsilon}$.
    
    In the last step, \cref{alg:1l_median_1_candidates} returns a set $C$ of all curves with up to $2\ell-2$ vertices from the grids, that contains one curve, denoted by $c^{\prime\prime}$ with same number of vertices as $c^\prime$ (recall that this is the curve guaranteed from \cref{lem:shortcutting_multiple}) and distance at most $\frac{\epsilon}{n} \Delta_l \leq \frac{\epsilon}{\lvert T^\prime \rvert} \cost{T^\prime}{c^\ast}$ between every corresponding pair of vertices of $c^\prime$ and $c^{\prime\prime}$. We conclude that $d_F(c^\prime, c^{\prime\prime}) \leq \frac{\epsilon}{\lvert T^\prime \rvert} \cost{T^\prime}{c^\ast}$. Also, recall that $d_F(\tau, c^\prime) \leq d_F(\tau, c^\ast)$ for $\tau \in T^\prime \setminus (T^\prime_1 \cup \dots \cup T^\prime_{2\ell-4})$. Further, $T^\prime$ contains at least $\frac{\lvert T^\prime \rvert}{2}$ curves with distance at most $\frac{2\cost{T^\prime}{c^\ast}}{\lvert T^\prime \rvert}$ to $c^\ast$, otherwise the cost of the remaining curves would exceed $\cost{T^\prime}{c^\ast}$, which is a contradiction, and since $\epsilon < \frac{1}{2}$ there is at least one curve $\sigma \in T^\prime\setminus (T^\prime_1 \cup \dots \cup T^\prime_{2\ell-4})$ with $d_F(\sigma, c^\prime) \leq d_F(\sigma, c^\ast) \leq \frac{2\cost{T^\prime}{c^\ast}}{\lvert T^\prime \rvert}$ by the pigeonhole principle. We can now bound the cost of $c^{\prime\prime}$ as follows:
    \begin{align*}
        \cost{T^\prime}{c^{\prime\prime}}  & ={} \sum_{\tau \in T^\prime} d_F(\tau, c^{\prime\prime}) \leq \sum_{\tau \in T^\prime \setminus (T^\prime_1 \cup \dots \cup T^\prime_{2\ell-4})} \left(d_F(\tau, c^\prime) + \frac{\epsilon}{\lvert T^\prime \rvert} \cost{T^\prime}{c^\ast}\right)\ + \\
        & \ \ \ \ \sum_{\tau \in (T^\prime_1 \cup \dots \cup T^\prime_{2\ell-4})} \left(d_F(\tau, c^\ast) + d_F(c^\ast, \sigma) + d_F(\sigma, c^{\prime}) + d_F(c^\prime, c^{\prime\prime}) \right) \\
        & \leq{} (1+\epsilon) \cost{T^\prime}{c^\ast} + \sum_{\tau \in (T^\prime_1 \cup \dots \cup T^\prime_{2\ell-4})} \left((2+2+\epsilon)\frac{\cost{T^\prime}{c^\ast}}{\lvert T^\prime \rvert} \right) \\
        & \leq{} \cost{T^\prime}{c^\ast} + \epsilon \cost{T^\prime}{c^\ast} + 5\epsilon \cost{T^\prime}{c^\ast} = (1 + 6 \epsilon) \cost{T^\prime}{c^\ast}.
    \end{align*}
    
    \textbf{Case 2:} $\lvert T^\prime_{c^\ast} \setminus B_{c} \rvert \leq 2 \epsilon \lvert T^\prime_{c^\ast} \rvert$
    
    Again, we distinguish two cases:
    \\
    
    \textbf{Case 2.1:} $d_F(c,c^\ast) \leq 4\epsilon \frac{\cost{T^\prime}{c^\ast}}{\lvert T^\prime \rvert}$
    
    We can easily bound the cost of $c$:
    \begin{align*}
        \cost{T^\prime}{c} \leq \sum_{\tau \in T^\prime} (d_F(\tau, c^\ast) + d_F(c^\ast, c)) \leq (1+4\epsilon) \cost{T^\prime}{c^\ast}.
    \end{align*}
        
    \textbf{Case 2.2:} $d_F(c,c^\ast) > 4\epsilon \frac{\cost{T^\prime}{c^\ast}}{\lvert T^\prime \rvert}$
    
    Recall that $\lvert T^\prime_{c^\ast} \rvert \geq (1-\epsilon^2) \lvert T^\prime \rvert$. We have
    \begin{align*}
        \lvert T^\prime \setminus B_{c} \rvert & \leq{} \lvert T^\prime \setminus T^\prime_{c^\ast} \rvert + 2\epsilon \lvert T^\prime_{c^\ast} \rvert = \lvert T^\prime \rvert - (1-2\epsilon) \lvert T^\prime_{c^\ast} \rvert \leq \lvert T^\prime \rvert - (1-2\epsilon)(1-\epsilon^2) \lvert T^\prime \rvert \\
        & = (2\epsilon + \epsilon^2 - 2\epsilon^3) \lvert T^\prime \rvert < \frac{1}{3} \lvert T^\prime \rvert.
    \end{align*}
    
    Hence, $\lvert T^\prime \cap B_{c} \rvert \geq (1-2\epsilon-\epsilon^2 + 2 \epsilon^3) \lvert T^\prime \rvert > \frac{2}{3} \lvert T^\prime \rvert$. Assume we assign all curves to $c$ instead of to $c^\ast$. For $\tau \in T^\prime \cap B_{c}$ we now have decrease in cost $d_F(\tau, c^\ast) - d_F(\tau, c)$, which can be bounded as follows:
    \begin{align*}
        d_F(\tau, c^\ast) - d_F(\tau, c) & \geq{} d_F(\tau, c^\ast) - \epsilon \frac{\cost{T^\prime}{c^\ast}}{\lvert T^\prime \rvert} \geq d_F(c,c^\ast) - d_F(\tau, c) - \epsilon \frac{\cost{T^\prime}{c^\ast}}{\lvert T^\prime \rvert} \\
        & \geq d_F(c, c^\ast) - 2\epsilon \frac{\cost{T^\prime}{c^\ast}}{\lvert T^\prime \rvert} > \frac{1}{2} d_F(c,c^\ast).
    \end{align*}
    
    For $\tau \in T^\prime \setminus B_{c}$ we have an increase in cost $d_F(\tau, c) - d_F(\tau, c^\ast) \leq d_F(c, c^\ast)$. Let the overall increase in cost be denoted by $\alpha$, which can be bounded as follows:
    \begin{align*}
        \alpha < \lvert T^\prime \setminus B_{c} \rvert \cdot d_F(c, c^\ast) - \lvert T^\prime \cap B_{c} \rvert \cdot \frac{d_F(c,c^\ast)}{2}.
    \end{align*}
    
    By the fact that $\lvert T^\prime \setminus B_{c} \rvert < \frac{1}{2} \lvert T^\prime \cap B_{c} \rvert$ for our choice of $\epsilon$, we conclude that $\alpha < 0$, which is a contradiction because $c^\ast$ is an optimal $(1,\ell)$-median for $T^\prime$. Therefore, case 2.2 cannot occur. Rescaling $\epsilon$ by $\frac{1}{6}$ proves the claim.
\end{proof}

We analyse the worst-case running-time of \cref{alg:1l_median_1_candidates} and the number of candidates it returns.

\begin{theorem}
    \label{theo:1l_median_1_running_time}
    \cref{alg:1l_median_1_candidates} has running-time and returns number of candidates $2^{O\left(\frac{(-\ln(\delta))^2 \cdot \beta}{\epsilon^2} + \log(m) \right)}$.
\end{theorem}
\begin{proof}
    The sample $S$ has size $O\left(\frac{-\ln(\delta) \cdot \beta}{\epsilon}\right)$ and sampling it takes time $O\left(\frac{-\ln(\delta) \cdot \beta}{\epsilon}\right)$. Let $n_S = \lvert S \rvert$. The for-loop runs \[\binom{n_S}{\frac{n_S}{2 \beta}} \in 2^{O\left(\frac{n_S}{2\beta} \log n_S\right)} \subset 2^{O\left(\frac{(-\ln(\delta))^2 \cdot \beta}{\epsilon^2}\right)}\] times. In each iteration, we run \cref{alg:1l_median_6}, taking time $O(m^2 \log(m) (-\ln^2 \delta) + m^3 \log m)$ (\conferre \cref{coro:1l_median_34}), we compute the cost of the returned curve with respect to $S^\prime$, taking time $O\left(\frac{-\ln(\delta)}{\epsilon} \cdot m \log(m)\right)$, and per curve in $S^\prime$ we build up to $m$ grids of size \[\left(\frac{\frac{(1+\epsilon)\Delta}{\epsilon}}{\frac{2\epsilon 2 \delta n \Delta}{n\sqrt{d} 4 \lvert S \rvert}}\right)^d = \left(\frac{\sqrt{d} \lvert S \rvert (1+\epsilon)}{\epsilon^2 \delta}\right)^d \in O\left(\frac{\beta^d(-\ln \delta)^d}{\epsilon^{3d}\delta^d}\right)\] each. \cref{alg:1l_median_1_candidates} then enumerates all combinations of $2\ell-2$ points from up to $\lvert S^\prime \rvert \cdot m$ grids, resulting in \[O\left(\frac{m^{2\ell-2} \beta^{2\ell d-2d+2\ell-2} (-\ln \delta)^{2\ell d-2d+2\ell-2}}{\epsilon^{6\ell d-6d+2\ell-2}\delta^{2\ell d-2d}}\right)\] candidates per iteration of the for-loop. Thus, \cref{alg:1l_median_1_candidates} computes $O\left(\poly{m,\beta,\delta^{-1},\epsilon^{-1}}\right)$ candidates per iteration of the for-loop and enumeration also takes time $O\left(\poly{m,\beta,\delta^{-1},\epsilon^{-1}}\right)$ per iteration of the for-loop.
    
    All in all, we have running-time and number of candidates $2^{O\left(\frac{(-\ln(\delta))^2 \cdot \beta}{\epsilon^2} + \log(m)\right)}$.
\end{proof}

\section{\texorpdfstring{$(1+\epsilon)$-Approximation}{Approximation-Scheme} for \texorpdfstring{$(k,\ell)$-}{(k,l)-}Median}
\label{sec:k-median}

We generalize the algorithm of \citet{DBLP:journals/talg/AckermannBS10} in the following way: instead of drawing a uniform sample and running a problem-specific algorithm on this sample in the candidate phase, we only run a problem-specific ``plugin''-algorithm in the candidate phase, thus dropping the framework around the sampling property. We think that the problem-specific algorithms used in \citep{DBLP:journals/talg/AckermannBS10} do not fulfill the role of a plugin, since parts of the problem-specific operations, e.g. the uniform sampling, remain in the main algorithm. Here, we separate the problem-specific operations from the main algorithm: any algorithm can serve as plugin, if it is able to return candidates for a cluster that takes a constant fraction of the input, where the fraction is an input-parameter of the algorithm and some approximation factor is guaranteed (w.h.p.). The calls to the candidate-finder plugin do not even need to be independent (stochastically), allowing adaptive algorithms.

Now, let $\mathcal{X} = (X,\rho)$ be an arbitrary space, where $X$ is any non-empty (ground-)set and $\rho \colon X \times X \rightarrow \mathbb{R}_{\geq 0}$ is a distance function (not necessarily a metric). We introduce a generalized definition of $k$-median clustering. Let the medians be restricted to come from a predefined subset $Y \subseteq X$. 

\begin{definition}[generalized $k$-median]
    \label{def:generalized_k_median}
    The generalized $k$-median clustering problem is defined as follows, where $k \in \mathbb{N}$ is a fixed (constant) parameter of the problem: given a finite and non-empty set $Z \subseteq X$, compute a set $C$ of $k$ elements from $Y$, such that $\cost{Z}{C} = \sum\limits_{z \in Z} \min\limits_{c \in C} \rho(z,c)$ is minimal.
\end{definition}

The following algorithm, \cref{alg:kl_median}, can approximate every $k$-median problem compatible with \cref{def:generalized_k_median}, when provided with such a problem-specific plugin-algorithm for computing candidates. In particular, it can approximate the $(k,\ell)$-median problem for polygonal curves under the Fr\'echet distance, when provided with \cref{alg:1l_median_3_candidates} or \cref{alg:1l_median_1_candidates}. Then, we have $X = \mathbb{X}^d$, $Y \subseteq \mathbb{X}^d_{\ell} \subseteq \mathbb{X}^d = X$ and $Z \subseteq \mathbb{X}^d_m \subseteq \mathbb{X}^d = X$. Note that the algorithm computes a bicriteria approximation, that is, the solution is approximated in terms of the cost \emph{and} the number of vertices of the center curves, i.e., the centers come from $\mathbb{X}^d_{2\ell-2}$.

\cref{alg:kl_median} has several parameters. The first parameter $C$ is the set of centers found yet and $\kappa$ is the number of centers yet to be found. The following parameters concern only the plugin-algorithm used within the algorithm: $\beta$ determines the size (in terms of a fraction of the input) of the smallest cluster for which an approximate median can be computed, $\delta$ determines the probability of failure of the plugin-algorithm and $\epsilon$ determines the approximation factor of the plugin-algorithm.

\cref{alg:kl_median} works as follows: If it has already computed some centers (and there are still centers to compute) it does \emph{pruning}: some clusters might be too small for the plugin-algorithm to compute approximate medians for them. \cref{alg:kl_median} then calls itself recursively with only half of the input: the elements with larger distances to the centers yet found. This way the small clusters will eventually take a larger fraction of the input and can be found in the \emph{candidate} phase. In this phase \cref{alg:kl_median} calls its plugin and for each candidate that the plugin returned, it calls itself recursively: adding the candidate at hand to the set of centers yet found and decrementing $\kappa$ by one. Eventually, all combinations of computed candidates are evaluated against the original input and the centers that together evaluated best are returned. 

\begin{algorithm}[H]
\caption{Recursive Approximation-Scheme for $k$-Median Clustering\label{alg:kl_median}}
    \begin{algorithmic}[1]
        \Procedure{$k$-Median}{$T, C, \kappa, \beta, \delta, \epsilon$}
            \If{$\kappa = 0$}
                \State \Return $C$
            \EndIf
            \Comment{\textbf{Pruning Phase}}
            \If{$C \neq \emptyset$}
                \State $P \gets $ set of $\left\lfloor\frac{\lvert T \rvert}{2}\right\rfloor$ elements $\tau \in T$, such that $\min\limits_{c \in C} \rho(\tau, c) \leq \min\limits_{c \in C} \rho(\sigma, c)$ for each $\sigma \in T \setminus P$
                \State $C^\prime \gets$ $k$-\textsc{Median}($T \setminus P, C, \kappa, \beta, \delta, \epsilon$)
            \Else
                \State $C^\prime \gets \emptyset$
            \EndIf
            \Comment{\textbf{Candidate Phase}}
            \State $K \gets 1$-\textsc{Median}-\textsc{Candidates}$(T,\beta,\delta/k,\epsilon)$
            \For{$c \in K$}
                \State $C_c \gets k$-\textsc{Median}$(T, C \cup \{c\}, \kappa-1, \beta, \delta, \epsilon)$
            \EndFor
            \State $\mathcal{C} \gets \{ C^\prime \} \cup \bigcup\limits_{c \in K} \{ C_c \}$
            \State \Return $\argmin\limits_{C \in \mathcal{C}} \cost{T}{C}$
        \EndProcedure
    \end{algorithmic}
\end{algorithm}

The quality of approximation and worst-case running-time of \cref{alg:kl_median} is stated in the following two theorems, which we prove further below. The proofs are adaptations of corresponding proofs in~\cite{DBLP:journals/talg/AckermannBS10}. We provide them for the sake of completeness.

\begin{theorem}
    \label{theo:kl_median_guarantee}
    Let $T = \{\tau_1, \dots, \tau_n\} \subseteq X$, $\alpha \in [1, \infty)$ and $1$-\textsc{Median}-\textsc{Candidates} be an algorithm that, given three parameters $\beta \in [1, \infty)$, $\delta, \epsilon \in (0,1)$ and a set $T \subseteq X$, returns with probability at least $1-\delta$ an $(\alpha + \epsilon)$-approximate $1$-median for any $T^\prime \subseteq T$, if $\lvert T^\prime \rvert \geq \frac{1}{\beta} \lvert T \rvert$. 
    
    \cref{alg:kl_median} called with parameters $(T,\emptyset,k,\beta,\delta, \epsilon)$, where $\beta \in (2k, \infty)$ and $\delta, \epsilon \in (0,1)$, returns with probability at least $1-\delta$ a set $C = \{c_1, \dots, c_k\}$ with $\cost{T}{C} \leq (1+\frac{4k^2}{\beta-2k})(\alpha + \epsilon) \cost{T}{C^\ast}$, where $C^\ast$ is an optimal set of $k$ medians for $T$.
\end{theorem}

\begin{theorem}
    \label{theo:kl_median_running_time}
    Let $T_1(n, \beta, \delta, \epsilon)$ denote the worst-case running-time of $1$-\textsc{Median}-\textsc{Candidates} for an arbitrary input-set $T$ with $\lvert T \rvert = n$ and let $C(n, \beta, \delta, \epsilon)$ denote the maximum number of candidates it returns. Also, let $T_d$ denote the worst-case running-time needed to compute $d$ for an input element and a candidate. 
    
    If $T_1$ and $C$ are non-decreasing in $n$, \cref{alg:kl_median} has running-time $O(C(n,\beta,\delta,\epsilon)^{k+2} \cdot n \cdot T_d + C(n,\beta,\delta,\epsilon)^{k+1} \cdot T_1(n, \beta, \delta, \epsilon))$.
\end{theorem}
Now we state our main results, which follow from \cref{theo:1l_median_3_guarantee,theo:1l_median_3_running_time}, respectively \cref{theo:1l_median_1_guarantee,theo:1l_median_1_running_time}, and \cref{theo:kl_median_guarantee,theo:kl_median_running_time}.

\begin{corollary}
    \label{coro:3_approx}
    Given two parameters $\delta, \epsilon \in (0,1)$ and a set $T \subset \mathbb{X}^d_m$ of polygonal curves, \cref{alg:kl_median} endowed with \cref{alg:1l_median_3_candidates} as $1$-\textsc{Median}-\textsc{Candidates} and run with parameters $(T,\emptyset,k,\frac{20k^2}{\epsilon}+2k,\delta, \epsilon/5)$ returns with probability at least $1-\delta$ a set $C \subset \mathbb{X}^d_{2\ell-2}$ that is a $(3+\epsilon)$-approximate solution to the $(k,\ell)$-median for $T$. \cref{alg:kl_median} then has running-time $n \cdot 2^{O\left(\frac{(-\ln(\delta))^2}{\epsilon^3} + \log(m) \right)}$.
\end{corollary}

\begin{corollary}
    \label{coro:1_approx}
    Given two parameters $\delta \in (0,1), \epsilon \in (0, 0.158] $ and a set $T \subset \mathbb{X}^d_m$ of polygonal curves, \cref{alg:kl_median} endowed with \cref{alg:1l_median_1_candidates} as $1$-\textsc{Median}-\textsc{Candidates} and run with parameters $(T,\emptyset,k,\frac{12k^2}{\epsilon}+2k,\delta, \epsilon/3)$ returns with probability at least $1-\delta$ a set $C \subset \mathbb{X}^d_{2\ell-2}$ that is a $(1+\epsilon)$-approximate solution to the $(k,\ell)$-median for $T$. \cref{alg:kl_median} then has running-time $n \cdot 2^{O\left(\frac{(-\ln(\delta))^2}{\epsilon^3} + \log(m)\right)}$.
\end{corollary}

The following proof is an adaption of \cite[Theorem 2.2 - Theorem 2.5]{DBLP:journals/talg/AckermannBS10}. 

\begin{proof}[Proof of Theorem~\ref{theo:kl_median_guarantee}]
    For $k=1$, the claim trivially holds. We now distinguish two cases. In the first case the principle of the proof is presented in all its detail. In the second case we only show how to generalize the first case to $k > 2$.\\
    
    \textbf{Case 1:} $k=2$

    Let $C^\ast = \{c^\ast_1, c^\ast_2 \}$ be an optimal set of $k$ medians for $T$ with clusters $T^\ast_1$ and $T^\ast_2$, respectively, that form a partition of $T$. For the sake of simplicity, assume that $n$ is a power of $2$ and w.l.o.g. assume that $\lvert T^\ast_1 \rvert \geq \frac{1}{2} \lvert T \rvert > \frac{1}{\beta} \lvert T \rvert$. Let $C_1$ be the set of candidates returned by $1$-\textsc{Median}-\textsc{Candidates} in the initial call. With probability at least $1-\delta/k$, there is a $c_1 \in C_1$  with $\cost{T^\ast_1}{c_1} \leq (\alpha+\epsilon) \cost{T^\ast_1}{c^\ast_1}$. We distinguish two cases:
    
    \paragraph{Case 1.1:} There exists a recursive call with parameters $(T^\prime,\{c_1\},1,\beta,\delta,\epsilon)$ and $\lvert T^\ast_2 \cap T^\prime \rvert \geq \frac{1}{\beta} \lvert T^\prime \rvert$.
    
    First, we assume that $T^\prime$ is the maximum cardinality input with $\lvert T^\ast_2 \cap T^\prime \rvert \geq \frac{1}{\beta} \lvert T^\prime \rvert$, occurring in a recursive call of the algorithm. Let $C_2$ be the set of candidates returned by $1$-\textsc{Median}-\textsc{Candidates} in this call. With probability at least $1-\delta/k$, there is a $c_2 \in C_2$  with $\cost{T^\ast_2 \cap T^\prime}{c_2} \leq (\alpha+\epsilon) \cost{T^\ast_2 \cap T^\prime}{\widetilde{c}_2}$, where $\widetilde{c}_2$ is an optimal median for $T^\ast_2 \cap T^\prime$. 
    
    Let $P$ be the set of elements of $T$ removed in the $m \in \mathbb{N}$, $m \leq \log_2(n)$, pruning phases between obtaining $c_1$ and $c_2$. Without loss of generality we assume that $P \neq \emptyset$. For $i \in \{1, \dots, m\}$, let $P_i$ be the elements removed in the $i$\textsuperscript{th} (in the order of the recursive calls occurring) pruning phase. Note that the $P_i$ are pairwise disjoint, we have that $P = \cup_{i=1}^t P_i$ and $\lvert P_i \rvert = \frac{n}{2^{i}}$. Since $T = T^\ast_1 \uplus (T^\ast_2 \cap T^\prime) \uplus (T^\ast_2 \cap P)$, we have
    \begin{align}
        \cost{T}{\{c_1, c_2\}} \leq \cost{T^\ast_1}{c_1} + \cost{T^\ast_2 \cap T^\prime}{c_2} + \cost{T^\ast_2 \cap P}{c_1} \label{eq:cluster2cost} \tag{I}.
    \end{align}
    Our aim is now to prove that the number of elements wrongly assigned to $c_1$, i.e., $T^\ast_2 \cap P$, is small and further, that their cost is a fraction of the cost of the elements correctly assigned to $c_1$, i.e., $T^\ast_1$.
    
    We define $R_0 = T$ and for $i \in \{1, \dots, m\}$ we define $R_i = R_{i-1} \setminus P_i$. The $R_i$ are the elements remaining after the $i$\textsuperscript{th} pruning phase. Note that by definition $\lvert R_i \rvert = \frac{n}{2^i} = \lvert P_i \rvert$. Since $R_{m} = T^\prime $ is the maximum cardinality input, with $\lvert T^\ast_2 \cap T^\prime \rvert \geq \frac{1}{\beta} \lvert T^\prime \rvert$, we have that $\lvert T^\ast_2 \cap R_i \rvert < \frac{1}{\beta} \lvert R_i \rvert$ for all $i \in \{1, \dots, m-1\}$. Also, for each $i \in \{1, \dots, m\}$ we have $P_i \subseteq R_{i-1}$, therefore 
    \begin{align}
        \lvert T^\ast_2 \cap P_i \rvert \leq \lvert T^\ast_2 \cap R_{i-1} \rvert < \frac{1}{\beta} \lvert R_{i-1} \rvert = \frac{2}{\beta} \frac{n}{2^i} \label{eq:wrongassignedupper} \tag{II}
    \end{align}
    and as immediate consequence
    \begin{align}
        \lvert T^\ast_1 \cap P_i \rvert = \lvert P_i \rvert - \lvert T^\ast_2 \cap P_i \rvert > \lvert P_i \rvert - \frac{1}{\beta} \lvert R_{i-1} \rvert = \left(1-\frac{2}{\beta}\right) \frac{n}{2^i}. \label{eq:correctassignedlower} \tag{III}
    \end{align}
    This tells us that mainly the elements of $T^\ast_1$ are removed in the pruning phase and only very few elements of $T^\ast_2$. By definition, we have for all $i \in \{1, \dots, m-1\}$, $\sigma \in P_i$ and $\tau \in P_{i+1}$ that $\rho(\sigma, c_1) \leq \rho(\tau, c_1)$, hence \[\frac{1}{\lvert T^\ast_2 \cap P_i \rvert} \cost{T^\ast_2 \cap P_i}{c_1} \leq \frac{1}{\lvert T^\ast_1 \cap P_{i+1} \rvert} \cost{T^\ast_1 \cap P_{i+1}}{c_1}.\] Combining this inequality with \cref{eq:wrongassignedupper,eq:correctassignedlower} we obtain for $i \in \{1, \dots, m-1\}$:
    \begin{align*}
        & \frac{\beta 2^i}{2n} \cost{T^\ast_2 \cap P_i}{c_1} < \frac{2^{i+1}}{(1-2/\beta)n} \cost{T^\ast_1 \cap P_{i+1}}{c_1} \\
        \Leftrightarrow & \cost{T^\ast_2 \cap P_i}{c_1} < \frac{2^{i+1} 2n}{(1-2/\beta)n\beta 2^i} \cost{T^\ast_1 \cap P_{i+1}}{c_1} = \frac{4}{(\beta-2)} \cost{T^\ast_1 \cap P_{i+1}}{c_1}. \label{eq:costwrongassigned} \tag{IV}
    \end{align*}
    We still need such a bound for $i=m$. Since $\lvert R_m \rvert = \lvert P_m \rvert$ and also $R_m \subseteq R_{m-1}$ we can use \cref{eq:wrongassignedupper} to obtain:
    \begin{align}
        \lvert T^\ast_1 \cap R_m \rvert = \lvert R_m \rvert - \lvert T^\ast_2 \cap R_m \rvert \geq \lvert R_m \rvert - \lvert T^\ast_2 \cap R_{m-1} \rvert > \left(1-\frac{2}{\beta}\right)\frac{n}{2^m}. \label{eq:remaininglower} \tag{V}
    \end{align}
    Also, we have for all $\sigma \in P_m$ and $\tau \in R_m$ that $\rho(\sigma, c_1) \leq \rho(\tau, c_1)$ by definition, thus \[\frac{1}{\lvert T^\ast_2 \cap P_m \rvert} \cost{T^\ast_2 \cap P_m}{c_1} \leq \frac{1}{\lvert T^\ast_1 \cap R_{m} \rvert} \cost{T^\ast_1 \cap R_m}{c_1}.\] We combine this inequality with \cref{eq:wrongassignedupper} and \cref{eq:remaininglower} and obtain:
    \begin{align*}
        & \frac{\beta 2^m}{2n} \cost{T^\ast_2 \cap P_m}{c_1} < \frac{2^m2n}{(1-2/\beta)n\beta2^m} \cost{T^\ast_1 \cap R_m}{c_1} \\
        \Leftrightarrow & \cost{T^\ast_2 \cap P_m}{c_1} < \frac{2}{(\beta-2)} \cost{T^\ast_1 \cap R_m}{c_1}. \label{eq:remainingcostwrongassigned} \tag{VI}
    \end{align*}
    We are now ready to bound the cost of the elements of $T^\ast_2$ wrongly assigned to $c_1$. Combining \cref{eq:costwrongassigned} and \cref{eq:remainingcostwrongassigned} yields:
    \begin{align*}
        \cost{T^\ast_2 \cap P}{c_1} & ={} \sum_{i=1}^m \cost{T^\ast_2 \cap P_i}{c_1} < \frac{4}{\beta-2} \sum_{i=1}^{m-1} \cost{T^\ast_1 \cap P_{i+1}}{c_1} + \frac{2}{\beta-2} \cost{T^\ast_1 \cap R_m}{c_1} \\
        & < \frac{4}{\beta-2} \cost{T^\ast_1}{c_1}.
    \end{align*}
    Here, the last inequality holds, because $P_2, \dots, P_m$ and $R_m$ are pairwise disjoint.
    Also, we have 
    \begin{align*}
        \cost{T^\ast_2 \cap T^\prime}{c_2} \leq (\alpha+\epsilon) \cost{T^\ast_2 \cap T^\prime}{\widetilde{c_2}} \leq (\alpha+\epsilon) \cost{T^\ast_2 \cap T^\prime}{c^\ast_2} \leq (\alpha+\epsilon) \cost{T^\ast_2}{c^\ast_2}.
    \end{align*}
    Finally, using \cref{eq:cluster2cost} and a union bound, with probability at least $1-\delta$ the following holds:
    \begin{align*}
        \cost{T}{\{c_1,c_2\}} & < (\alpha+\epsilon) \cost{T^\ast_1}{c^\ast_1} + (\alpha+\epsilon) \cost{T^\ast_2}{c^\ast_2} + \frac{4}{\beta-2} (\alpha+\epsilon) \cost{T^\ast_1}{c^\ast_1} \\
        & < \left(1+\frac{4}{\beta-2}\right) (\alpha + \epsilon) \cost{T}{C^\ast} = \left(1+\frac{4k}{k\beta-2k}\right) (\alpha + \epsilon) \cost{T}{C^\ast} \\
        & \leq{} \left(1+\frac{4k^2}{\beta-2k}\right) (\alpha + \epsilon) \cost{T}{C^\ast}.
    \end{align*}
    
    \textbf{Case 1.2:} For all recursive calls with parameters $(T^\prime,\{c_1\},1,\beta,\delta,\epsilon)$ it holds that $\lvert T^\ast_2 \cap T^\prime \rvert < \frac{1}{\beta} \lvert T^\prime \rvert$.
    
    After $\log_2(n)$ pruning phases we end up with a singleton $\{\sigma\} = T^\prime$ as input set. Since $\lvert T^\ast_2 \cap T^\prime \rvert < \frac{1}{\beta} \lvert T^\prime \rvert$, it must be that $0 = \lvert T^\ast_2 \cap T^\prime \rvert < \frac{1}{\beta} \lvert T^\prime \rvert = \frac{1}{\beta} < 1$ and thus $\sigma \in T^\ast_1$.
    
    Let $C_2$ be the set of candidates returned by $1$-\textsc{Median}-\textsc{Candidates} in this call. With probability at least $1-\delta/k$ there is a $c_2 \in C_2$ with $\cost{\{\sigma\}}{c_2} \leq (\alpha + \epsilon) \cost{\{\sigma\}}{\widetilde{c}_2} \leq (\alpha + \epsilon) \cost{\{\sigma\}}{c^\ast_1}$, where $\widetilde{c}_2$ is an optimal median for $\{ \sigma \}$. Since $\cost{T^\ast_2 \cap P}{c_1}$ is bounded as in Case 1.1, by a union bound we have with probability at least $1- \delta$:
    \begin{align*}
        \cost{T}{\{c_1, c_2\}} & \leq{} \cost{T^\ast_1 \setminus \{ \sigma \}}{c_1} + \cost{T^\ast_2 \cap P}{c_1} + \cost{\{\sigma\}}{c_2} \\
        & \leq (\alpha + \epsilon) \cost{T^\ast_1}{c^\ast_1} + \cost{T^\ast_2 \cap P}{c_1} \\
        & \leq \left(1+\frac{4}{\beta-2}\right) (\alpha + \epsilon) \cost{T}{C^\ast} \\
        & \leq \left(1+\frac{4k^2}{\beta-2k}\right) (\alpha + \epsilon) \cost{T}{C^\ast}.
    \end{align*}
    
    \textbf{Case 2:} $k > 2$
    
    We only prove the generalization of Case 1.1 to $k > 2$, the remainder of the proof is analogous to the Case 1. For the sake of brevity, for $i \in \mathbb{N}$, we define $[i] = \{1, \dots, i\}$. Let $C^\ast = \{c^\ast_1, \dots, c^\ast_k \}$ be an optimal set of $k$ medians for $T$ with clusters $T^\ast_1, \dots, T^\ast_k$, respectively, that form a partition of $T$. For the sake of simplicity, assume that $n$ is a power of $2$ and w.l.o.g. assume $\lvert T^\ast_1 \rvert \geq \dots \geq \lvert T^\ast_k \rvert$. For $i \in [k]$ and $j \in [k] \setminus [i]$ we define $T^\ast_{i,j} = \uplus_{t=i}^j T^\ast_t$. 
    
    Let $\mathcal{T}_0 = T$ and let $(\mathcal{T}_j = \mathcal{T}_{j-1} \setminus \mathcal{P}_j)_{j=1}^m$ be the sequence of input sets in the recursive calls of the $m \in \mathbb{N}$, $m \leq \log_2(n)$, pruning phases, where $\mathcal{P}_j$ is the set of elements removed in the $j$\textsuperscript{th} (in the order of the recursive calls occurring) pruning phase. Let $\mathcal{T} = \{\mathcal{T}_0\} \cup \{ \mathcal{T}_j \mid j \in [m] \}$. For $i \in [k]$, let $T_i$ be the maximum cardinality set in $\mathcal{T}$, with $\lvert T^\ast_i \cap T_i \rvert \geq \frac{1}{\beta} \lvert T_i \rvert$. Note that by assumption and since $\beta > 2k$, $T_1 = T$ must hold and also $T_j \subset T_i$ for $j \in [k] \setminus [i]$. 
    
    Using a union bound, with probability at least $1-\delta$, for each $i \in [k]$ the call of $1$-\textsc{Median}-\textsc{Candidates} with input $T_i$ yields a candidate $c_i$ with 
    \begin{align*}
        \cost{T^\ast_i \cap T_i}{c_i} \leq (\alpha+\epsilon) \cost{T^\ast_i \cap T_i}{\widetilde{c}_i} \leq (\alpha+\epsilon) \cost{T^\ast_i \cap T_i}{c^\ast_i} \leq (\alpha+\epsilon) \cost{T^\ast_i}{c^\ast_i}, \label{eq:candidatekapprox} \tag{I}
    \end{align*}
    where $\widetilde{c}_i$ is an optimal $1$-median for $T^\ast_i \cap T_i$. Let $C = \{c_1, \dots, c_k\}$ be the set of these candidates and for $i \in [k-1]$, let $P_i = T_{i} \setminus T_{i+1}$ denote the set of elements of $T$ removed by the pruning phases between obtaining $c_{i}$ and $c_{i+1}$. Note that the $P_i$ are pairwise disjoint.
    
    By definition, the sets \[T^\ast_1 \cap T_1, \dots, T^\ast_k \cap T_k, T^\ast_{2,k} \cap P_1, \dots, T^\ast_{k,k} \cap P_{k-1}\] form a partition of $T$, therefore
    \begin{align*}
        \cost{T}{\{ c_1, \dots, c_k \}} & \leq{} \sum_{i=1}^k \cost{T^\ast_i \cap T_i}{c_i} + \sum_{i=1}^{k-1} \cost{T^\ast_{i+1,k} \cap P_{i}}{\{ c_1, \dots, c_{i} \}} \\
        & \leq{} (\alpha + \epsilon) \sum_{i=1}^k \cost{T^\ast_i}{c^\ast_i} + \sum_{i=1}^{k-1} \cost{T^\ast_{i+1,k} \cap P_{i}}{\{ c_1, \dots, c_{i} \}}. \label{eq:costkcluster} \tag{II}
    \end{align*}
    Now, it only remains to bound the cost of the wrongly assigned elements of $T^\ast_{i+1,k}$. For $i \in [k]$, let $n_i = \lvert T_i \rvert$ and w.l.o.g. assume that $P_i \neq \emptyset$ for each $i \in [k-1]$. Each $P_i$ is the disjoint union $\uplus_{j=1}^{m_i} P_{i,j}$ of $m_i \in \mathbb{N}$ sets of elements of $T$ removed in the interim pruning phases and it holds that $\lvert P_{i,j} \rvert = \frac{n_{i}}{2^j}$. We now prove for each $i \in [k-1]$ and $j \in [m_i]$ that $P_i$ contains a large number of elements from $T^\ast_{1,i}$ and only a few elements from $T^\ast_{i+1,k}$. 
    
    For $i \in [k-1]$, we define $R_{i,0} = T_i$ and for $j \in [m_i]$ we define $R_{i,j} = R_{i,j-1} \setminus P_{i,j}$. By definition, $\lvert R_{i,j} \rvert = \frac{n_i}{2^j} = \lvert P_{i,j} \rvert$, $R_{i,j_1} \supset R_{i,j_2}$ for each $j_1 \in [m_i]$ and $j_2 \in [m_i] \setminus [j_1]$, also $R_{i,m_i} = T_{i+1}$. Thus, $\lvert T^\ast_t \cap R_{i,j} \rvert < \frac{1}{\beta} \lvert R_{i,j} \rvert$ for all $i \in [k-1], j \in [m_i]$ and $t \in [k] \setminus [i]$. As immediate consequence we obtain $\lvert T^\ast_{i+1,k} \cap R_{i,j} \rvert \leq \frac{k}{\beta} \lvert R_{i,j} \rvert$. Since $P_{i,j} \subseteq R_{i,j-1}$ for all $i \in [k-1]$ and $j \in [m_i]$, we have
    \begin{align*}
        \lvert T_{i+1,k} \cap P_{i,j} \rvert \leq \lvert T_{i+1,k} \cap R_{i,j-1} \rvert \leq \frac{k}{\beta} \lvert R_{i,j-1} \rvert = \frac{2k}{\beta} \frac{n_i}{2^j}, \label{eq:wronglyassignedupperk} \tag{III}
    \end{align*}
    which immediately yields
    \begin{align*}
        \lvert T_{1,i} \cap P_{i,j} \rvert = \lvert P_{i,j} \rvert - \lvert T_{i+1,k} \cap P_{i,j} \rvert \geq \left(1-\frac{2k}{\beta}\right) \frac{n_i}{2^j}. \label{eq:correctlyassignedlowerk} \tag{IV}
    \end{align*}
    Now, by definition we know that for all $i \in [k-1]$, $j \in [m_i] \setminus \{m_i\}$, $\sigma \in P_{i,j}$ and $\tau \in P_{i,j+1}$ that $\min\limits_{c \in \{c_1, \dots, c_{i}\}}\rho(\sigma, c) \leq \min\limits_{c \in \{c_1, \dots, c_{i}\}}\rho(\tau, c)$. Thus,
    \begin{align*}
        \frac{\cost{T^\ast_{i+1,k} \cap P_{i,j}}{\{c_1, \dots, c_i\}}}{\lvert T^\ast_{i+1,k} \cap P_{i,j} \rvert} \leq \frac{\cost{T^\ast_{1,i} \cap P_{i,j+1}}{\{c_1, \dots, c_i\}}}{\lvert T^\ast_{1,i} \cap P_{i,j+1} \rvert}.
    \end{align*}
    Combining this inequality with \cref{eq:wronglyassignedupperk,eq:correctlyassignedlowerk} yields for $i \in [k-1]$ and $j \in [m_i] \setminus \{m_i\}$:
    \begin{align*}
        & \frac{\beta 2^j}{2kn_i} \cost{T^\ast_{i+1,k} \cap P_{i,j}}{\{c_1, \dots, c_i\}}  \leq \frac{2^{j+1}}{(1-\frac{2k}{\beta})n_i} \cost{T^\ast_{1,i} \cap P_{i,j+1}}{\{c_1, \dots, c_i\}} \\
        \Leftrightarrow & \cost{T^\ast_{i+1,k} \cap P_{i,j}}{\{c_1, \dots, c_i\}} \leq \frac{4k}{\beta-2k} \cost{T^\ast_{1,i} \cap P_{i,j+1}}{\{c_1, \dots, c_i\}} \label{eq:costwronglyassignedk} \tag{V}
    \end{align*}
    For each $i \in [k-1]$ we still need an upper bound on $\cost{T^\ast_{i+1,k} \cap P_{i,m_i}}{\{c_1, \dots, c_i\}}$. Since $\lvert R_{i,m_i} \rvert = \lvert P_{i,m_i} \rvert$ and also $R_{i,m_i} \subseteq R_{i,m_i-1}$ we can use \cref{eq:wronglyassignedupperk} to obtain 
    \begin{align*}
        \lvert T^\ast_{1,i} \cap R_{i,m_i} \rvert = \lvert R_{i,m_i} \lvert - \lvert T^\ast_{i+1,k} \cap R_{i,m_i} \rvert \geq \lvert R_{i,m_i} \lvert - \lvert T^\ast_{i+1,k} \cap R_{i,m_i-1} \rvert > \left(1-\frac{2k}{\beta}\right)\frac{n_i}{2^{m_i}}. \label{eq:remaininglowerk} \tag{VI}
    \end{align*}
    By definition we also know that for all $i \in [k-1]$, $\sigma \in P_{i,m_i}$ and $\tau \in R_{i,m_i}$ that $\min\limits_{c \in \{c_1, \dots, c_{i}\}}\rho(\sigma, c) \leq \min\limits_{c \in \{c_1, \dots, c_{i}\}}\rho(\tau, c)$. Thus, \[\frac{\cost{T^\ast_{i+1,k} \cap P_{i,m_i}}{\{c_1, \dots, c_i \}}}{\lvert T^\ast_{i+1,k} \cap P_{i,m_i} \rvert} \leq \frac{\cost{T^\ast_{1,i} \cap R_{i,m_i}}{\{c_1, \dots, c_i \}}}{\lvert T^\ast_{1,i} \cap R_{i,m_i} \rvert}.\] Combining this inequality with \cref{eq:wronglyassignedupperk,eq:remaininglowerk} yields:
    \begin{align*}
        & \frac{\beta 2^{m_i}}{2kn_i} \cost{T^\ast_{i+1,k} \cap P_{i,m_i}}{\{c_1, \dots, c_i \}} < \frac{2^{m_i}}{(1-\frac{2k}{\beta})n_i} \cost{T^\ast_{1,i} \cap R_{i,m_i}}{\{c_1, \dots, c_i \}} \\
        \Leftrightarrow & \cost{T^\ast_{i+1,k} \cap P_{i,m_i}}{\{c_1, \dots, c_i \}} < \frac{2k}{\beta-2k} \cost{T^\ast_{1,i} \cap R_{i,m_i}}{\{c_1, \dots, c_i \}}. \label{eq:remainingcostwronglyassignedk} \tag{VII}
    \end{align*}
    We can now give the following bound, combining \cref{eq:costwronglyassignedk,eq:remainingcostwronglyassignedk}, for each $i \in [k-1]$:
    \begin{align*}
        \cost{T^\ast_{i+1,k} \cap P_i}{\{c_1, \dots, c_i\}} & ={} \sum_{j=1}^{m_i} \cost{T^\ast_{i+1,k} \cap P_{i,j}}{\{c_1, \dots, c_i\}} \\
        & < \sum_{j=1}^{m_i-1} \frac{4k}{\beta-2k} \cost{T^\ast_{1,i} \cap P_{i,j+1}}{\{c_1, \dots, c_i\}} \\
        & \ \ \ + \frac{2k}{\beta-2k} \cost{T^\ast_{1,i} \cap R_{i,m_i}}{\{c_1, \dots, c_i \}} \\
        & < \frac{4k}{\beta-2k} \cost{T^\ast_{1,i} \cap T_i}{\{c_1, \dots, c_i \}}. \label{eq:costwronglyupper} \tag{VIII}
    \end{align*}
    Here, the last inequality holds, because $P_{i,2}, \dots, P_{i,m_i}$ and $R_{i,m_i}$ are pairwise disjoint subsets of $T_i$.
    
    Now, we plug this bound into \cref{eq:costkcluster}. Note that $T^\ast_j \cap T_i \subseteq T^\ast_j \cap T_j$ for each $i \in [k]$ and $j \in [i]$ by definition. We obtain:
    \begin{align*}
        \cost{T}{\{ c_1, \dots, c_k \}} & \leq{} (\alpha + \epsilon) \sum_{i=1}^k \cost{T^\ast_i}{c^\ast_i} + \sum_{i=1}^{k-1} \cost{T^\ast_{i+1,k} \cap P_{i}}{\{ c_1, \dots, c_{i} \}} \\
        & < (\alpha + \epsilon) \sum_{i=1}^k \cost{T^\ast_i}{c^\ast_i} + \frac{4k}{\beta-2k} \sum_{i=1}^{k-1} \cost{T^\ast_{1,i} \cap T_i}{\{c_1, \dots, c_i \}} \\
        & \leq{} (\alpha + \epsilon) \sum_{i=1}^k \cost{T^\ast_i}{c^\ast_i} + \frac{4k}{\beta-2k} \sum_{i=1}^{k-1} \sum_{t=1}^i \cost{T^\ast_{t} \cap T_i}{c_t} \\
        & \leq{} (\alpha + \epsilon) \sum_{i=1}^k \cost{T^\ast_i}{c^\ast_i} + \frac{4k}{\beta-2k} \sum_{i=1}^{k-1} \sum_{t=1}^i \cost{T^\ast_{t} \cap T_t}{c_t}\\
        & \leq (\alpha + \epsilon) \sum_{i=1}^k \cost{T^\ast_i}{c^\ast_i} + \frac{4k^2}{\beta-2k} \sum_{i=1}^{k-1} \cost{T^\ast_{i} \cap T_i}{c_i} \\
        & \leq{} \left(1+\frac{4k^2}{\beta-2k}\right) (\alpha + \epsilon) \sum_{i=1}^k \cost{T^\ast_i}{c^\ast_i} = \left(1+\frac{4k^2}{\beta-2k}\right) (\alpha + \epsilon) \cost{T}{C^\ast}.
    \end{align*}
    The last inequality follows from \cref{eq:candidatekapprox}.
\end{proof}

The following analysis of the worst-case running-time of \cref{alg:1l_median_1_candidates} is a slight adaption of \cite[Theorem 2.8]{DBLP:journals/talg/AckermannBS10}, which is also provided for the sake of completeness.

\begin{proof}[Proof of Theorem~\ref{theo:kl_median_running_time}]
    For the sake of simplicity, we assume that $n$ is a power of $2$. 
    
    If $\kappa = 0$, \cref{alg:kl_median} has running-time $c_1 \in O(1)$ and if $\kappa \geq n$, \cref{alg:kl_median} has running-time $c_2 \cdot n \in O(n)$.
    
    Let $T(n, \kappa, \beta, \delta, \epsilon)$ denote the worst-case running-time of \cref{alg:kl_median} for input set $T$ with $\lvert T \rvert = n$. If $n > \kappa \geq 1$, \cref{alg:kl_median} has running-time at most $c_3 \cdot (n \cdot T_d + n) \in O(n \cdot T_d)$ to obtain $P$, $T(n/2, \kappa, \beta, \delta, \epsilon)$ for the recursive call in the pruning phase, $T_1(n, \beta, \delta, \epsilon)$ to obtain the candidates, $C(n, \beta, \delta, \epsilon) \cdot T(n, \kappa - 1, \beta, \delta, \epsilon)$ for the recursive calls in the candidate phase, one for each candidate, and $c_4 \cdot n \cdot T_d \cdot C(n, \beta, \delta, \epsilon) \in O(n \cdot T_d \cdot C(n, \beta, \delta, \epsilon))$ to eventually evaluate the candidate sets. Let $c = \max\{c_1, c_2, c_3, c_4\}$. We obtain the following recurrence relation:
    \begin{align*}
        T(n, \kappa, \beta, \delta, \epsilon) \leq
        \begin{cases}
            c & \text{if } \kappa = 0 \\
            c n & \text{if } \kappa \geq n \\
            C(n, \beta, \delta, \epsilon) \cdot T(n, \kappa -1, \beta, \delta, \epsilon) + T(n/2, \kappa, \beta, \delta, \epsilon) \\
            + T_1(n, \beta, \delta, \epsilon) + c n \cdot T_d \cdot C(n, \beta, \delta, \epsilon)) & \text{else}
        \end{cases}.
    \end{align*}
    
    Let $f(n, \beta, \delta, \epsilon) = \frac{1}{cn} \cdot T_1(n, \beta, \delta, \epsilon) +  T_d \cdot C(n, \beta, \delta, \epsilon)$.
    
    We prove that $T(n, \kappa, \beta, \delta, \epsilon) \leq c \cdot 4^{\kappa} \cdot C(n,\beta,\delta,\epsilon)^{\kappa+1} \cdot n \cdot f(n, \beta, \delta, \epsilon)$, by induction on $n, \kappa$.
    
    For $\kappa = 0$ we have $T(n,\kappa, \beta, \delta, \epsilon) \leq c \leq cn \leq  c \cdot 4^0 \cdot C(n,\beta,\delta,\epsilon) \cdot n \cdot f(n, \beta, \delta, \epsilon)$.
    
    For $\kappa \geq n$ we have $T(n,\kappa, \beta, \delta, \epsilon) \leq cn \leq c \cdot 4^{\kappa} \cdot C(n,\beta,\delta,\epsilon)^{\kappa+1} \cdot n \cdot f(n, \beta, \delta, \epsilon)$. 
    
    Now, let $n > \kappa \geq 1$ and assume the claim holds for $T(n^\prime,\kappa^\prime, \beta, \delta, \epsilon)$, for each $\kappa^\prime \in \{0, \dots, \kappa-1\}$ and $n^\prime \in \{1, \dots, n-1\}$. We have:
    \begin{align*}
        T(n, \kappa, \beta, \delta, \epsilon) & \leq{} C(n, \beta, \delta, \epsilon) \cdot T(n, \kappa -1, \beta, \delta, \epsilon) + T(n/2, \kappa, \beta, \delta, \epsilon) \\
        & \ \ \ + T_1(n, \beta, \delta, \epsilon) + cn \cdot T_d \cdot C(n, \beta, \delta, \epsilon) \\
        & \leq{} C(n, \beta, \delta, \epsilon) \cdot c \cdot 4^{\kappa-1} \cdot C(n,\beta,\delta,\epsilon)^{\kappa} \cdot n \cdot f(n, \beta, \delta, \epsilon) \\
        & \ \ \ + c \cdot 4^{\kappa} \cdot C(n/2,\beta,\delta,\epsilon)^{\kappa+1} \cdot \frac{n}{2} \cdot f(n/2, \beta, \delta, \epsilon) \\
        & \ \ \ + cn \cdot f(n, \beta, \delta, \epsilon) \\
        & \leq{} \left( \frac{1}{4} + \frac{1}{2} + \frac{1}{4^\kappa C(n,\beta,\delta,\epsilon)^{\kappa+1}}\right) c \cdot 4^{\kappa} \cdot C(n,\beta,\delta,\epsilon)^{\kappa+1} \cdot n \cdot f(n, \beta, \delta, \epsilon) \\
        & \leq{} c \cdot 4^{\kappa} \cdot C(n,\beta,\delta,\epsilon)^{\kappa+1} \cdot n \cdot f(n, \beta, \delta, \epsilon).
    \end{align*}
    
    The last inequality holds, because $\frac{1}{4^\kappa C(n,\beta,\delta,\epsilon)^{\kappa+1}} \leq \frac{1}{4}$, and the claim follows by induction.
\end{proof}

\section{Conclusion}

We have developed bi-criteria approximation algorithms for $(k,\ell)$-median clustering of polygonal curves under the Fr\'echet distance. While it showed to be relatively easy to obtain a good approximation where the centers have up to $2\ell$ vertices in reasonable time, a way to obtain good approximate centers with up to $\ell$ vertices in reasonable time is not in sight. This is due to the continuous Fr\'echet distance: the vertices of a median need not be anywhere near a vertex of an input-curve, resulting in a huge search-space. If we cover the whole search-space by, say grids, the worst-case running-time of the resulting algorithms become dependent on the arc-lengths of the input-curves edges, which is not acceptable. We note that $g$-coverability of the continuous Fréchet distance would imply the existence of sublinear size $\epsilon$-coresets for $(k,\ell)$-center clustering of polygonal curves under the Fr\'echet distance. It is an interesting open question, if the $g$-coverability holds for the continuous Fr\'echet distance.
In contrast to the doubling dimension, which was shown to be infinite even for curves of bounded complexity~\cite{DBLP:conf/soda/DriemelKS16}, the VC-dimension of metric balls under the continuous Fr\'echet distance is bounded in terms of the complexities $\ell$ and $m$ of the curves~\cite{Driemel19}. Whether this bound can be combined with the framework by Feldman and Langberg~\cite{Feldman2011} to achieve faster approximations for the $(k,\ell)$-median problem under the continuous Fr\'echet distance is an interesting open problem. The general relationship between the VC-dimension of range spaces derived from metric spaces and their doubling properties is a topic of ongoing research, see for example \citet{DBLP:conf/focs/HuangJLW18}.

\bibliography{bibliography}

\end{document}